\documentclass[preprint,3p,10pt]{elsarticle}
\biboptions{numbers,sort&compress}
\usepackage{xcolor} 
\usepackage{amsmath,amssymb,amsthm,bm}
\usepackage{mathrsfs}
\usepackage{geometry}
\usepackage{enumerate}
\usepackage[hidelinks]{hyperref}
\hypersetup{
  pdftitle={Vector rogue wave patterns associated with generalized Hermite and Okamoto
polynomials},
  pdfauthor={Hejiaqi Chen, Dongwei Wu, Chengfa Wu, Guangxiong Zhang}
}

\usepackage[capitalize,nameinlink]{cleveref}

\usepackage[capitalize,nameinlink]{cleveref}
\usepackage{breqn}
\usepackage{graphicx,subfigure,float,multirow}
\usepackage{pdflscape}
\usepackage{makecell}
\usepackage{booktabs}

\usepackage{ytableau}
\usepackage{tikz}

\def\R{\mathbb R}

\renewcommand{\i}{\mathrm{i}}
\newcommand{\abs}[1]{\left|{#1}\right|}

\newcounter{subeq}

\renewcommand{\Re}{\operatorname{Re}}
\renewcommand{\Im}{\operatorname{Im}}

\theoremstyle{definition}
	\newtheorem{definition}{Definition}[section]
    \newtheorem{remark}[definition]{Remark}
	
	\newtheorem{lemma}[definition]{Lemma}
	
	\newtheorem{theorem}[definition]{Theorem}

    \newtheorem{corollary}[definition]{Corollary}
\journal{xxx}

\begin{document}

\graphicspath{ {./figures/} }
\small

\begin{frontmatter}
\title{Vector rogue wave patterns associated with generalized Hermite and Okamoto polynomials}

\date{\today}

\author[szu_ias]{Hejiaqi Chen\fnref{fn1}}

\author[szu_ias]{Dongwei Wu\fnref{fn1}}

\author[szu_ias,szu_sms]{Chengfa Wu\corref{cor1}}
\ead{cfwu@szu.edu.cn}

\author[cityu_math]{Guangxiong Zhang\corref{cor1}}
\ead{gzhang74-c@my.cityu.edu.hk}

\cortext[cor1]{Corresponding authors.}
\fntext[fn1]{Hejiaqi Chen and Dongwei Wu contributed equally to this work.}
\affiliation[szu_ias]{organization={Institute for Advanced Study, Shenzhen University},
            city={Shenzhen},
            postcode={518060}, 
            country={People's Republic of China}}
\affiliation[szu_sms]{organization={School of Mathematical Sciences, Shenzhen University},
            city={Shenzhen},
            postcode={518060}, 
            country={People's Republic of China}}

\affiliation[cityu_math]{organization={Department of Mathematics, City University of Hong Kong},
                        addressline={83 Tat Chee Avenue, Kowloon Tong, Kowloon},
                        city={Hong Kong},
                        country={People's Republic of China}}

\begin{abstract}
We establish rogue wave patterns associated with the fourth Painlev\'e equation \( (\mathrm{P}_{\mathrm{IV}}) \) in the multi-component nonlinear Schr\"odinger and Hirota equations. The generalized Hermite and generalized Okamoto polynomials arise in representations of rational solutions of \(\mathrm{P}_{\mathrm{IV}}\), and we show that their roots determine two classes of rogue wave patterns when one of the internal parameters of rogue wave solutions is large. Specifically, the generalized Hermite polynomials arise from rogue wave solutions represented by Schur-polynomial determinants with consecutive indices, whereas the generalized Okamoto polynomials arise from analogous determinants with index jumps of three. Numerical examples for both equations agree with the predictions. 
\end{abstract}

\begin{keyword}
Multi-component NLS equation \sep Multi-component Hirota equation \sep rogue wave patterns \sep  generalized Hermite polynomials\sep generalized Okamoto polynomials\sep the fourth Painlev\'e equation
\end{keyword}
\end{frontmatter}

\section{Introduction}
Rogue waves have been studied extensively in integrable systems for several decades. Draper discussed unusually large ocean waves, referred to as ``freak'' waves \cite{draper1966freak}. Akhmediev, Ankiewicz and Taki later described rogue waves as waves that ``appear from nowhere and disappear without a trace'' \cite{akhmediev2009waves}. For higher-order NLS rogue waves, Ankiewicz, Kedziora and Akhmediev showed that suitable nonzero free parameters can produce a rogue wave triplet: three well-separated fundamental rogue waves whose peaks lie at the corners of a triangle in the \((x,t)\)-plane \cite{ankiewicz2011rogue}. Subsequent studies reported circular rogue wave clusters \cite{kedziora2011circular}, and Ohta and Yang showed that general high-order NLS rogue waves can form arrays of fundamental rogue waves appearing at different times and spatial positions \cite{ohta2012general}. These rogue wave patterns exhibit interesting behaviors, and systematic analysis of them may reveal profound links to special algebraic objects.

At the turn of the twentieth century, Painlev\'e, Gambier and Fuchs classified equations \(w''=F(z,w,w')\), with \(F\) rational in \(w\) and \(w'\) and locally analytic in \(z\), whose solutions have no movable branch points, a condition now known as the Painlevé property. They obtained fifty canonical equations with this property. Among these equations, forty-four can be reduced to linear equations, solved in terms of elliptic functions, or reduced to one of six new nonlinear equations. These six irreducible equations, now known as the Painlev\'e equations \(\mathrm{P}_{\mathrm{I}}\)--\(\mathrm{P}_{\mathrm{VI}}\), have generic solutions that define the Painlev\'e transcendents \cite{painleve1900,painleve1902,gambier1910,clarkson2019open}. However, for special values of parameters, \(\mathrm{P}_{\mathrm{II}}\)--\(\mathrm{P}_{\mathrm{VI}}\) possess hierarchies of rational solutions. For instance, rational solutions of \(\mathrm{P}_{\mathrm{II}}\) and \(\mathrm{P}_{\mathrm{III}}\)  are expressed through the Yablonskii--Vorob'ev and Umemura polynomials, respectively \cite{clarkson2003second,clarkson2003third}. Rational solutions of \(\mathrm{P}_{\mathrm{IV}}\) form three hierarchies: two of them are represented by generalized Hermite polynomials, whereas the remaining one is represented by generalized Okamoto polynomials \cite{clarkson2003fourth,NoumiYamada1999}.

Using large-parameter asymptotics, Yang and Yang showed that the root configurations of the Yablonskii--Vorob'ev polynomial hierarchy determine NLS rogue wave patterns \cite{yang2021rogue}. They subsequently extended this connection to several other integrable systems. For each system, an equation-dependent affine transformation maps the roots of the relevant Yablonskii--Vorob'ev polynomial to the space--time locations of the fundamental rogue waves \cite{yang2021universal}. More recently, rogue wave and lump patterns associated with Umemura polynomials,  and hence with \(\mathrm{P}_{\mathrm{III}}\), were established in \cite{yang2026roguewavelumppatternsassociated}. When several internal parameters become large, broader rogue wave patterns are described asymptotically by Adler--Moser polynomials through a dilation \cite{yang2024rogue}.

On the \(\mathrm{P}_{\mathrm{IV}}\) side, Yang and Yang identified rogue wave patterns associated with the Okamoto polynomial hierarchies \(Q_N^{[m]}(z)\) and \(R_N^{[m]}(z)\) \cite{yang2023rogue}, while patterns governed by generalized Okamoto polynomials were subsequently reported in \cite{zhang2025multi}. The generalized Hermite and generalized Okamoto polynomials, which form the two classical polynomial families associated with rational solutions of \(\mathrm{P}_{\mathrm{IV}}\), are both special cases of Wronskian--Hermite polynomials \cite{oblomkov1999monodromy,kajiwara1998determinant}. More generally, Lin and Ling identified certain rogue wave patterns of the vector NLS equation related to generalized mixed Adler--Moser polynomials, a broader family that includes the Wronskian--Hermite polynomials \cite{lin2025vector}.

Within the existing multi-block Schur-polynomial determinant framework, the present work explicitly identifies the index choices that realize rogue wave patterns associated with these two classical \(\mathrm{P}_{\mathrm{IV}}\) polynomial families and provides a unified large-parameter analysis of the multi-component NLS and Hirota equations. Specifically, consecutive determinant indices yield generalized Hermite patterns, whereas two index sequences with index jumps of three produce generalized Okamoto patterns in the two-component systems. An important feature of these two polynomial families is that all their roots are known to be simple \cite{Duran2020,NoumiYamada1999,roffelsen2025real}. Consequently, in the large-parameter regime considered here, we establish a correspondence in which each root of these polynomials determines one fundamental vector rogue wave, without requiring any assumption on the simplicity of roots. 

Specifically, we consider the \(M\)-component NLS equation
\begin{equation}
\label{vector NLS}
\mathrm{i}u_{j,t}+u_{j,xx}+\left(\sum_{k=1}^{M}\sigma_k|u_k|^2\right)u_j=0,\qquad j=1,2,\ldots,M,
\end{equation}
and the $M$-component Hirota equation
\begin{equation}
\label{eqn: vector-Hirota}
v_{j,t}=v_{j,xxx}-3\left(\sum_{k=1}^{M}c_k|v_k|^2\right)v_{j,x}-3v_j\left(\sum_{k=1}^{M}c_kv_k^{*}v_{k,x}\right),\qquad j=1,2,\ldots,M.
\end{equation}
Here, \(M\) is a positive integer, the asterisk denotes complex conjugation, \(\sigma_k=\pm1\), and \(c_k\in\mathbb{R}\). Both systems admit higher-order vector rogue wave solutions in Schur-polynomial determinant form. This common determinant structure provides the algebraic setting in which the two \(\mathrm{P}_{\mathrm{IV}}\) polynomial families enter the large-parameter asymptotics.
When the second free internal parameter in the Schur-polynomial representation becomes large, solutions with consecutive determinant indices produce patterns determined by roots of the generalized Hermite polynomials \(H_{d,n}(z)\). In the two-component systems, solutions indexed by \((N_1,N_2)\) produce patterns determined by roots of the generalized Okamoto polynomials \(Q_{N_1,N_2}(z)\). 

The remainder of this paper is organized as follows. Section~2 reviews the rational solutions of \(\mathrm{P}_{\mathrm{IV}}\) and introduces the generalized Hermite and generalized Okamoto polynomials, together with their Wronskian representations and properties of their roots. Section~3 recalls the Schur-polynomial determinant forms of higher-order rogue wave solutions for the multi-component NLS and Hirota equations and identifies the determinant index vectors corresponding to the two polynomial families. Section~4 establishes rogue wave patterns associated with the two polynomial families, derives explicit formulas for the centers of the resulting fundamental vector rogue waves, and presents numerical illustrations for the multi-component NLS and Hirota equations. Section~5 provides a proof of the main theorem. Finally, we summarize the main results in Section~6.

\section{Special polynomials associated with the fourth Painlevé equation}

\subsection{Rational solutions of the fourth Painlev\'e equation}
\label{subsec:rational-solutions-PIV}

The fourth Painlev\'e equation $(\mathrm{P}_{\mathrm{IV}})$ is
\begin{equation}
\label{eq:PIV}
    w''=\frac{(w')^2}{2w}+\frac{3}{2}w^3+4zw^2+2\left(z^2-\alpha\right)w+\frac{\beta}{w},\quad ':=\frac{d}{dz},
\end{equation}
where  $\alpha$, $\beta\in \mathbb C$ are constants. It has been shown \cite{Murata1985} that $\mathrm{P}_{\mathrm{IV}}$ admits a rational solution if and only if
\begin{equation}
\label{eq:PIV-rational-condition-I}
    \alpha=r,\qquad\beta=-2(1+2s-r)^2,\qquad r,s\in\mathbb{Z},
\end{equation}
or
\begin{equation}
\label{eq:PIV-rational-condition-II}
    \alpha=r,\qquad\beta=-\frac{2}{9}(1+6s-3r)^2,\qquad r,s\in\mathbb{Z}.
\end{equation}
For each  $(\alpha,\beta)$ satisfying either \eqref{eq:PIV-rational-condition-I} or \eqref{eq:PIV-rational-condition-II}, the rational solution of $\mathrm{P}_{\mathrm{IV}}$ is unique \cite{Murata1985}. Furthermore, these two parameter conditions can be decomposed into three parameter families \cite{Buckingham2020}

\begin{equation}
    \label{eq:PIV-parameter-families-I}
    \alpha_{r,s}^{(\mathrm{I})}=\pm(r+2s+1),\quad \beta_{r,s}^{(\mathrm{I})}=-2r^2,\quad r\in\mathbb Z_{>0},\quad s\in\mathbb Z_{\ge0},
\end{equation}
\begin{equation}
    \label{eq:PIV-parameter-families-II}
    \alpha_{r,s}^{(\mathrm{II})}=s-r,\quad \beta_{r,s}^{(\mathrm{II})}=-2(r+s+1)^2,\quad r,s\in\mathbb{Z}_{\ge 0},
\end{equation}
\begin{equation}
    \label{eq:PIV-parameter-families-oka}
    \alpha_{r,s}^{(\mathrm{Oka})}=r,\quad \beta_{r,s}^{(\mathrm{Oka})}=-2\left(2s-r+\frac{1}{3}\right)^2,\quad r,s\in\mathbb{Z},
\end{equation}
where $\mathbb Z_{\ge 0}$ denotes the set of all nonnegative integers. The families of rational solutions to $\mathrm{P}_{\mathrm{IV}}$ corresponding to \eqref{eq:PIV-parameter-families-I}, \eqref{eq:PIV-parameter-families-II}, and \eqref{eq:PIV-parameter-families-oka} can be referred to as the $\pm 1/z$, $-2z$, and $-2z/3$ hierarchies \cite{BassomClarksonHicks1995}, respectively. The three simplest rational solutions of $\mathrm{P}_{\mathrm{IV}}$ are
\begin{equation}
\label{eq:PIV-elementary-rational-solutions}
    w(z;\pm 2,-2)=\pm\frac{1}{z},\quad w(z;0,-2)=-2z, \quad w\left(z;0,-\frac{2}{9}\right)=-\frac{2}{3}z.
\end{equation}

For $(\alpha,\beta)$ satisfying  \eqref{eq:PIV-parameter-families-I} or \eqref{eq:PIV-parameter-families-II}, the rational solutions of $\mathrm{P}_{\mathrm{IV}}$ can be expressed in terms of generalized Hermite polynomials, whereas those satisfying \eqref{eq:PIV-parameter-families-oka} can be expressed in terms of generalized Okamoto polynomials. Thus, the generalized Hermite and generalized Okamoto polynomials are the two classes of special polynomials associated with  rational solutions of $\mathrm{P}_{\mathrm{IV}}$  \cite{kajiwara1998determinant,clarkson2003fourth,NoumiYamada1999}.

\subsection{Wronskian--Hermite polynomials}
The generalized Hermite and generalized Okamoto polynomials are special cases of the Wronskian--Hermite polynomials \cite{oblomkov1999monodromy} corresponding to two particular choices of the index vector \(\boldsymbol{\mu}\). Let the polynomials \(p_k(z)\), \(k\in\mathbb{Z}_{\geq 0}\), be defined by
\begin{equation}
\label{eq:generating-function-pk}
\sum_{k=0}^{\infty}p_k(z)\epsilon^k=\exp\left(z\epsilon+\epsilon^2\right),
\end{equation}
with \(p_k(z)\equiv0\) for \(k<0\). These polynomials satisfy $p_{k+1}'=p_k$.

For an ordered \(N\)-tuple of distinct nonnegative integers \(\boldsymbol{\mu}=(\mu_1,\mu_2,\ldots,\mu_N)\), define
\begin{equation}
\label{eq:Pmu}
P_{\boldsymbol{\mu}}(z)=C_{\boldsymbol{\mu}}\det_{1\leq i,j\leq N}\left[p_{\mu_i-j+1}(z)\right]=C_{\boldsymbol{\mu}}\operatorname{Wr}\left[p_{\mu_1}(z),p_{\mu_2}(z),\ldots,p_{\mu_N}(z)\right],
\end{equation}
where
\begin{equation}
\label{eq:Wronskian-normalization}
C_{\boldsymbol{\mu}}=\frac{\displaystyle\prod_{j=1}^{N}\mu_j!}{\displaystyle\prod_{1\leq i<j\leq N}(\mu_j-\mu_i)}.
\end{equation}
Then \(P_{\boldsymbol{\mu}}(z)\) is monic and
\begin{equation}
\label{eq:Pmu-degree}
\deg P_{\boldsymbol{\mu}}=\sum_{j=1}^{N}\mu_j-\frac{N(N-1)}{2}.
\end{equation}

\subsection{Generalized Hermite polynomials}
\label{subsec:generalized-Hermite-polynomials}

The generalized Hermite polynomials  can be characterized recursively by a pair of nonlinear differential--difference relations, together with appropriate initial conditions \cite{NoumiYamada1999,MasoeroRoffelsen2018}. They also admit  determinant  representations. To facilitate the analysis of rogue wave patterns, we employ a determinant representation that is compatible with the Schur-polynomial structure of the rogue wave solutions.

For positive integers \(d\) and \(n\), set \(N=n\) and take
\[
\boldsymbol{\mu}=(d,d+1,\ldots,d+n-1).
\]
Then, after transposing the determinant in \eqref{eq:Pmu}, the generalized Hermite polynomial \(H_{d,n}(z)=P_{\boldsymbol{\mu}}(z)\) can be written as \cite{kajiwara1998determinant}
\begin{equation}
\label{eq:generalized-Hermite-determinant-expanded}
H_{d,n}(z) = C_{\boldsymbol{\mu}}
\begin{vmatrix}
p_d(z) & p_{d+1}(z) & \cdots & p_{d+n-1}(z)\\ p_{d-1}(z) & p_d(z) & \cdots & p_{d+n-2}(z)\\ \vdots & \vdots & \ddots & \vdots\\ p_{d-n+1}(z) & p_{d-n+2}(z) & \cdots & p_d(z) \end{vmatrix} =  C_{\boldsymbol{\mu}} 
\operatorname{Wr}\left[p_d(z),p_{d+1}(z),\ldots,p_{d+n-1}(z)\right].
\end{equation}
For this index vector, \eqref{eq:Wronskian-normalization} reduces to
\begin{equation}
\label{eq:generalized-Hermite-normalization}
C_{\boldsymbol{\mu}} = \prod_{j=0}^{n-1}\frac{(d+j)!}{j!}.
\end{equation}
Thus, \(H_{d,n}(z)\) is a monic polynomial of degree \(dn\) with integer coefficients \cite{NoumiYamada1999,MasoeroRoffelsen2018,GrosuGrosu2021}. 

The properties of generalized Hermite polynomials have been investigated extensively in the literature \cite{felder2012zeros,bonneux2020coefficients,GrosuGrosu2021,grosu2021irreducibility}. In particular, the generalized Hermite polynomials satisfy  \cite{NoumiYamada1999,MasoeroRoffelsen2018}
\begin{equation}
    H_{d,n}(-z)=(-1)^{dn}H_{d,n}(z),\quad H_{d,n}(\mathrm{i}z)=\mathrm{i}^{dn}H_{n,d}(z).
\end{equation}
Thus, the set of roots of $H_{d,n}(z)$ is invariant under reflection about both the real and imaginary axes, while the root configurations of $H_{d,n}(z)$ and $H_{n,d}(z)$ are related by a rotation of $\pi/2$ in the complex plane. In addition, the generalized Hermite polynomials with consecutive  indices have only simple roots \cite{Duran2020}.

Figure~\ref{fig:zeros of the generalized Hermite polynomials} illustrates the root configurations of several generalized Hermite polynomials. The roots form approximately rectangular patterns with \(d\) rows and \(n\) columns. These configurations will be related to the rogue wave patterns in subsequent sections.

\begin{figure}[H]
\centering
\begin{tabular}{c cccc}
& $n=2$ & $n=3$ &
  $n=4$ & $n=5$ \\ 
\rotatebox{90}{\shortstack{$\quad d=1$}} &
\includegraphics[width=0.11\textwidth]{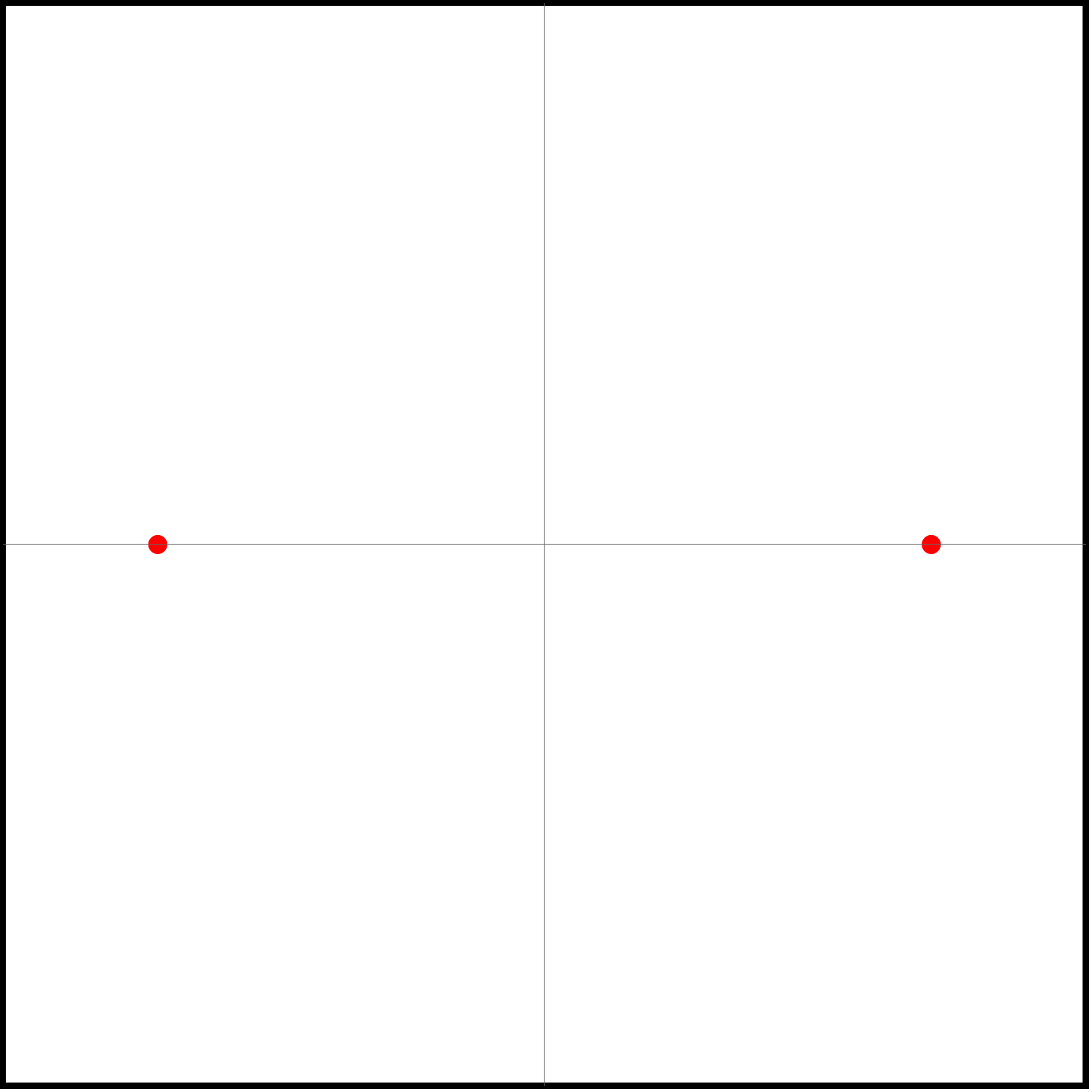} &
\includegraphics[width=0.11\textwidth]{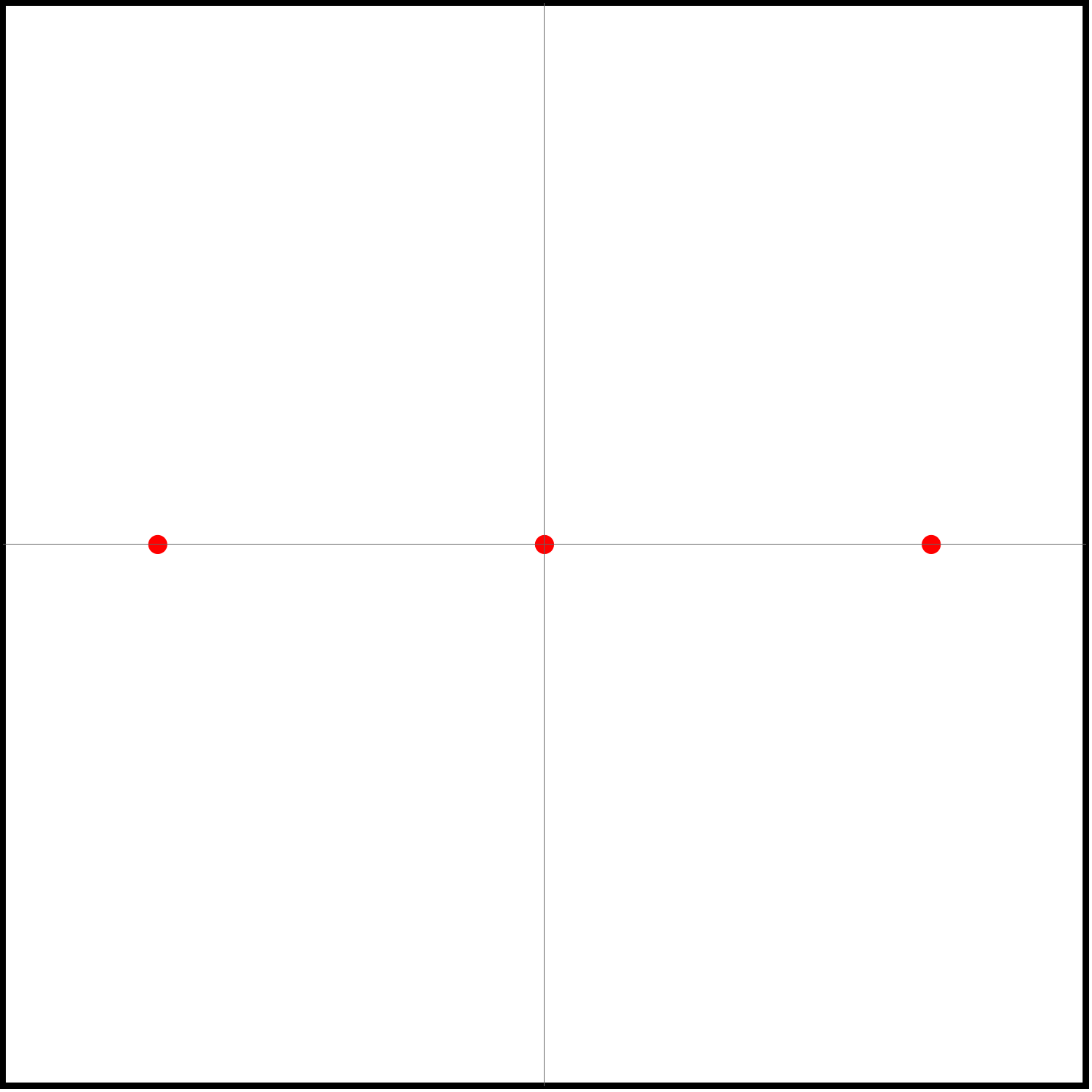} &
\includegraphics[width=0.11\textwidth]{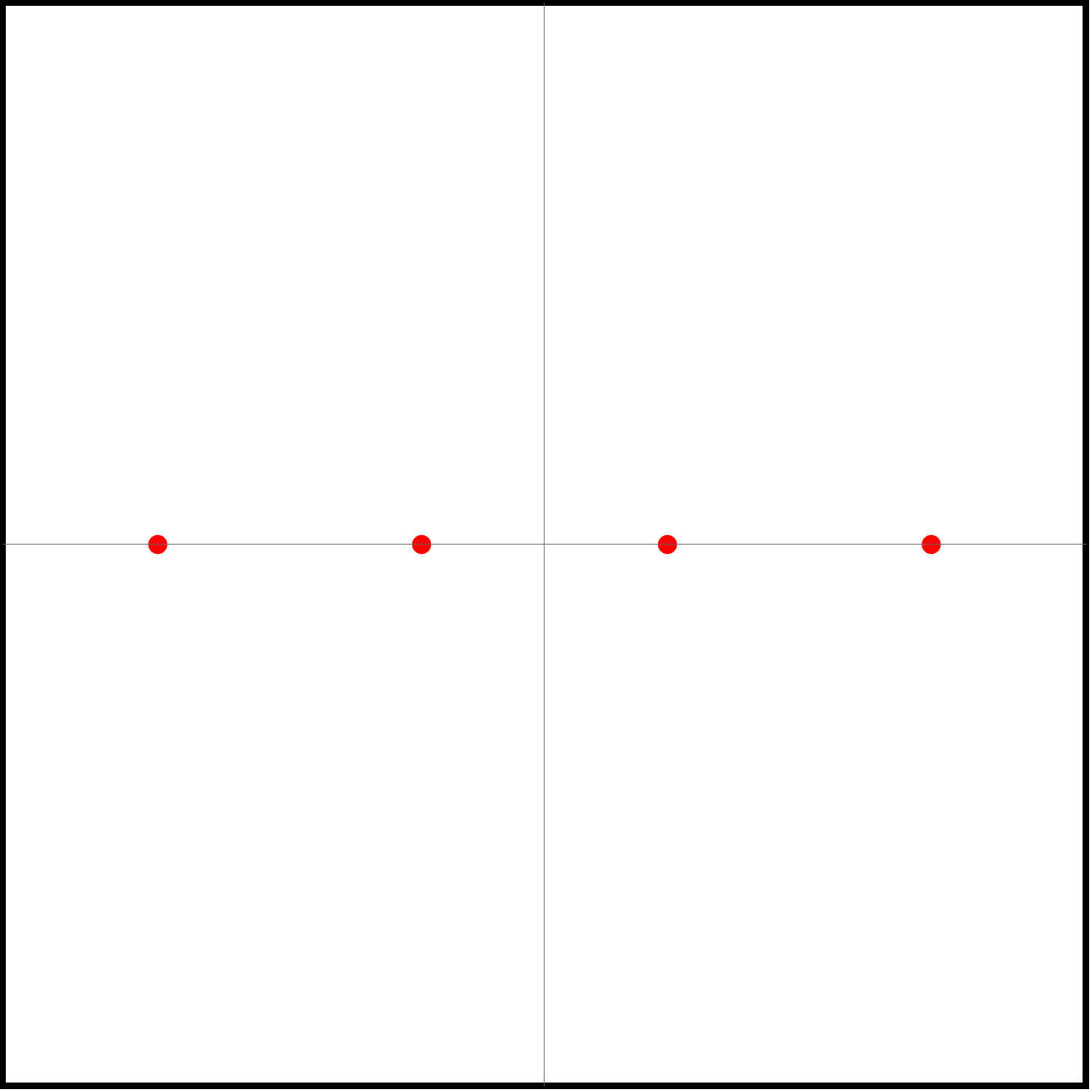} &
\includegraphics[width=0.11\textwidth]{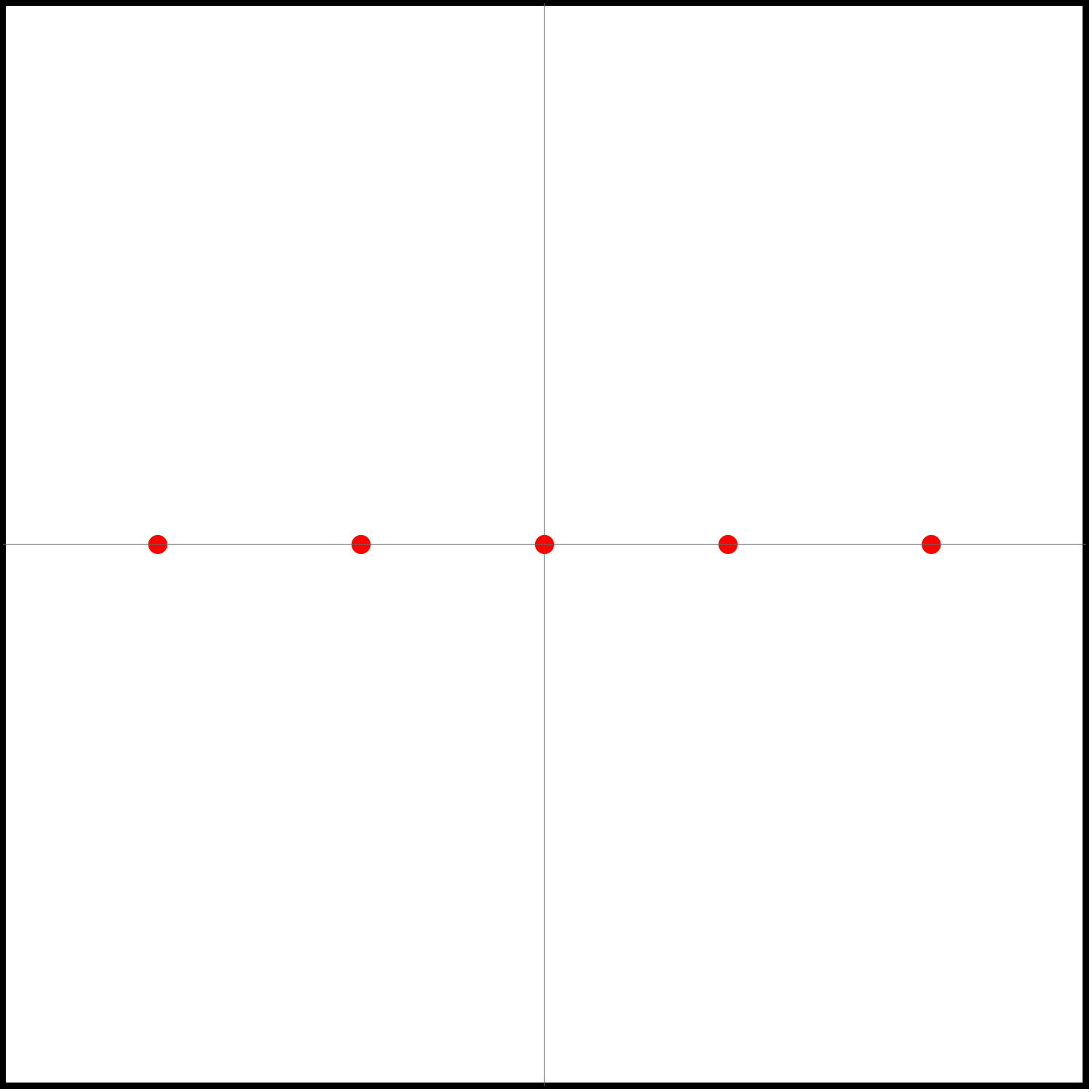} 
\\
\rotatebox{90}{\shortstack{$\quad d=2$}} &
\includegraphics[width=0.11\textwidth]{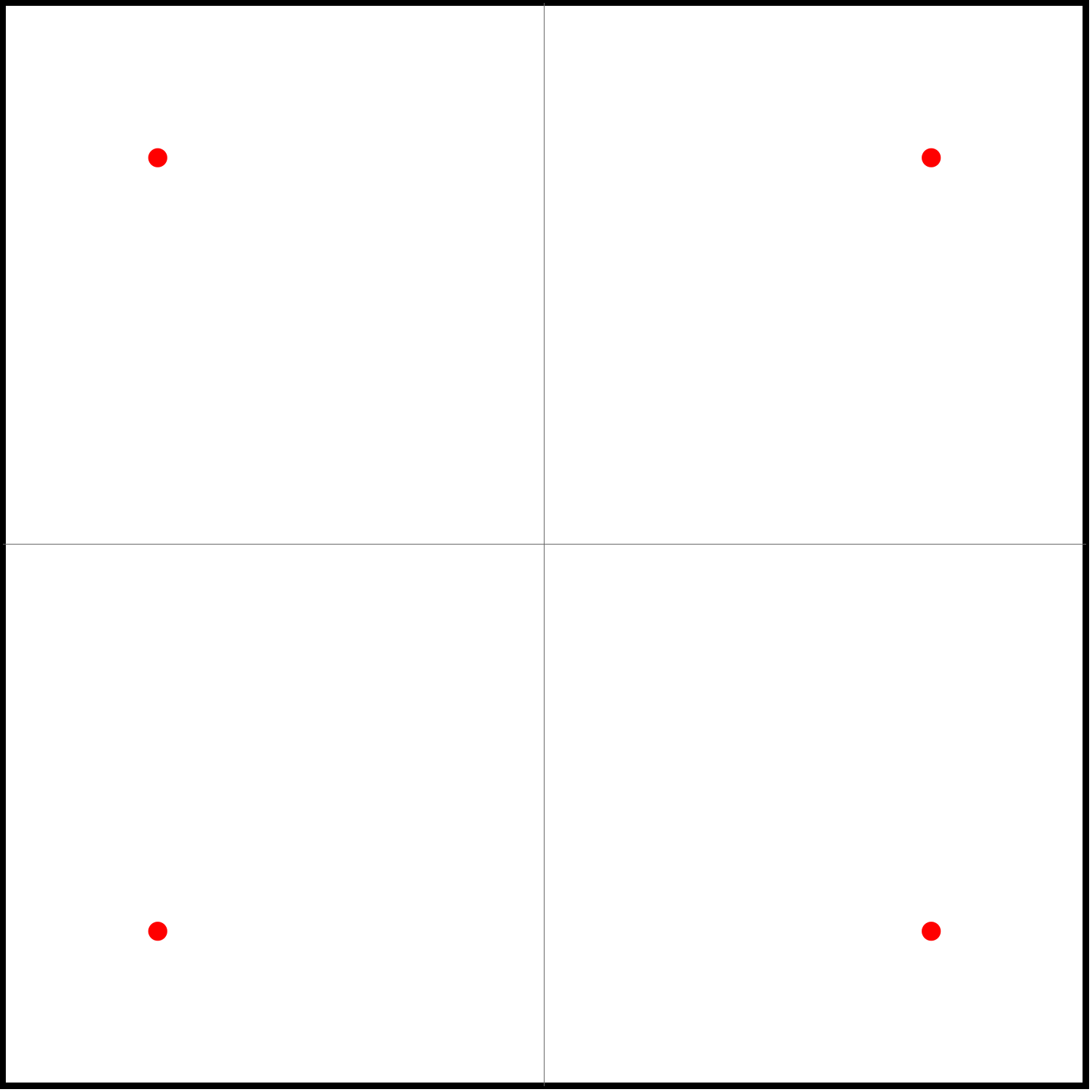} &
\includegraphics[width=0.11\textwidth]{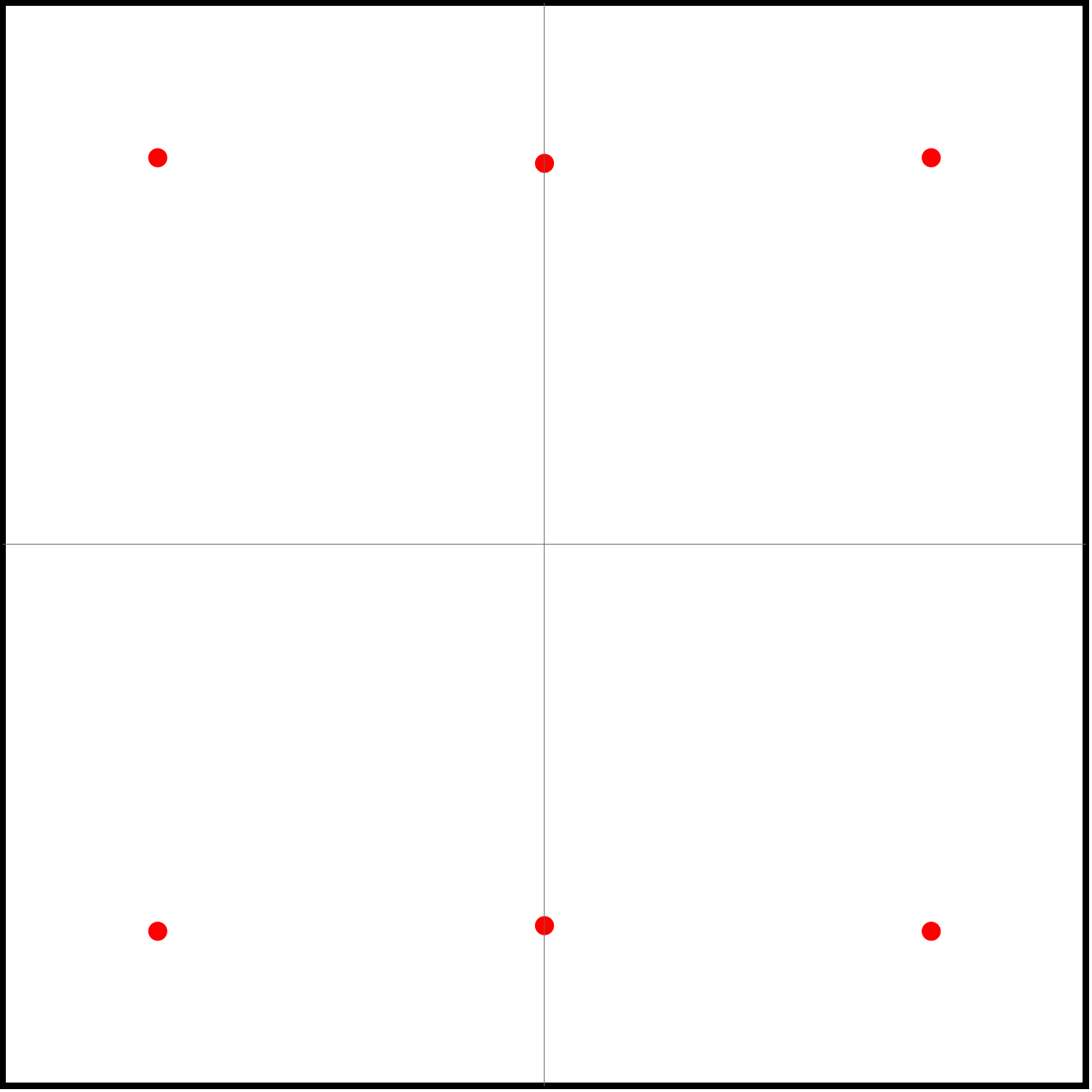} &
\includegraphics[width=0.11\textwidth]{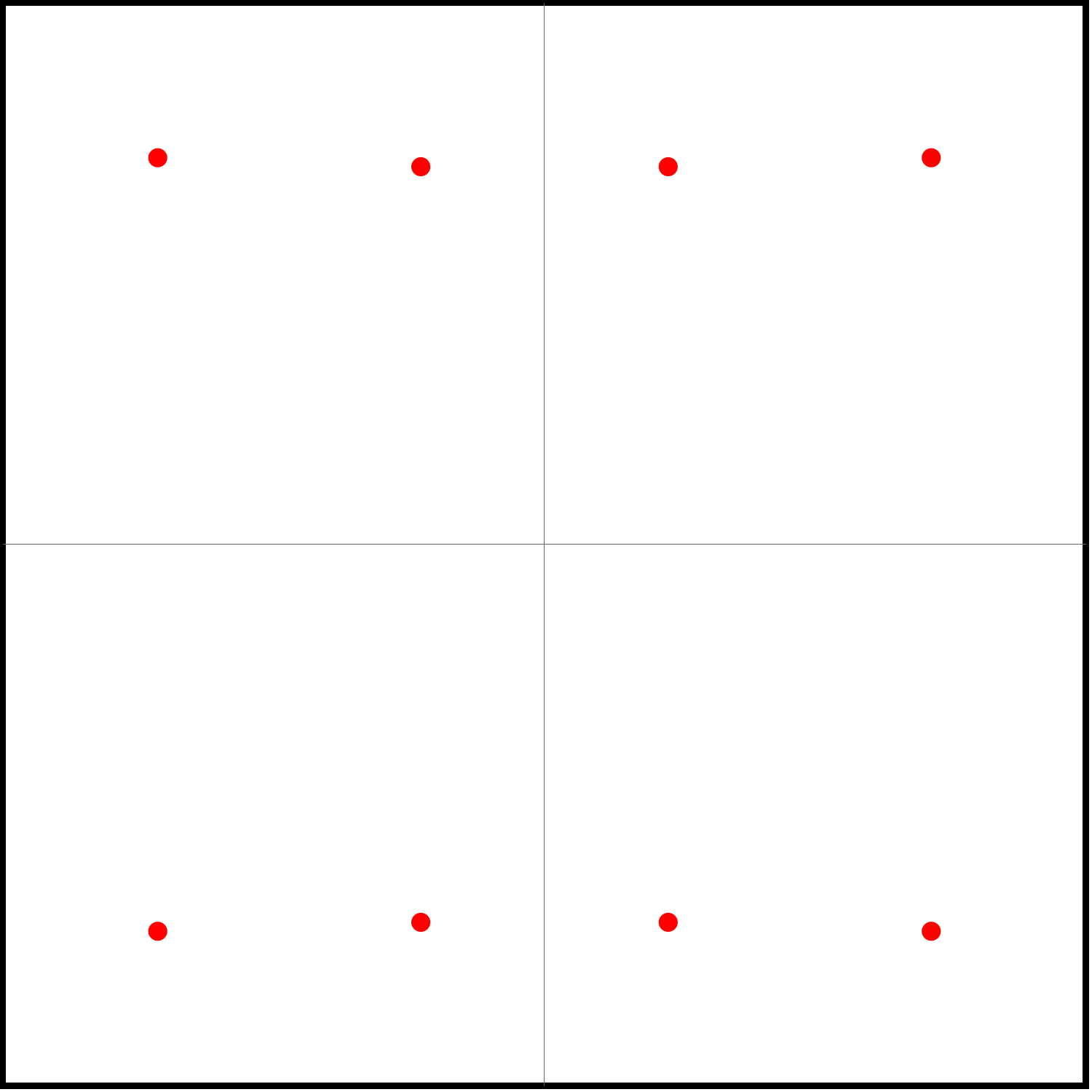} &
\includegraphics[width=0.11\textwidth]{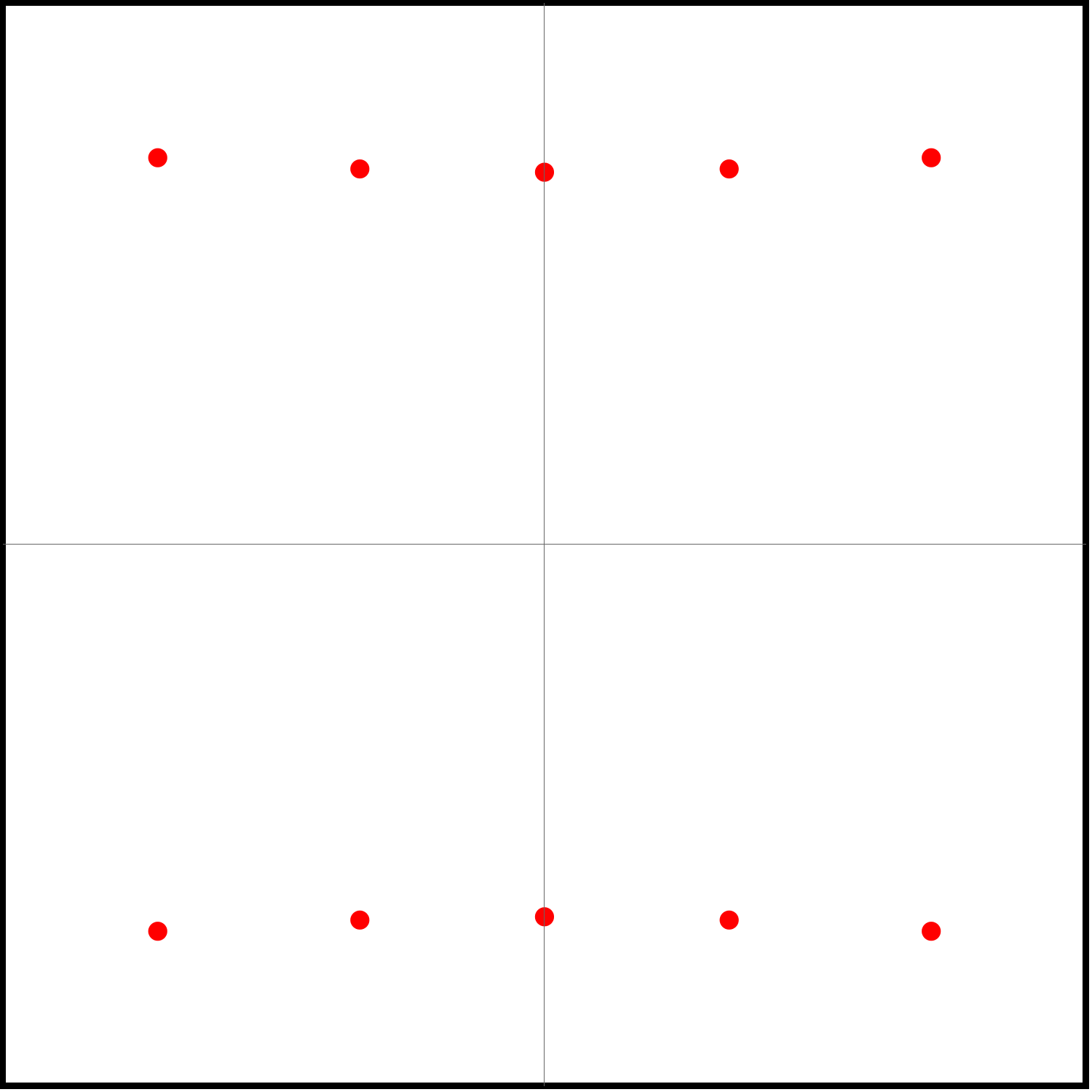} 
\\
\rotatebox{90}{\shortstack{$\quad d=3$}} &
\includegraphics[width=0.11\textwidth]{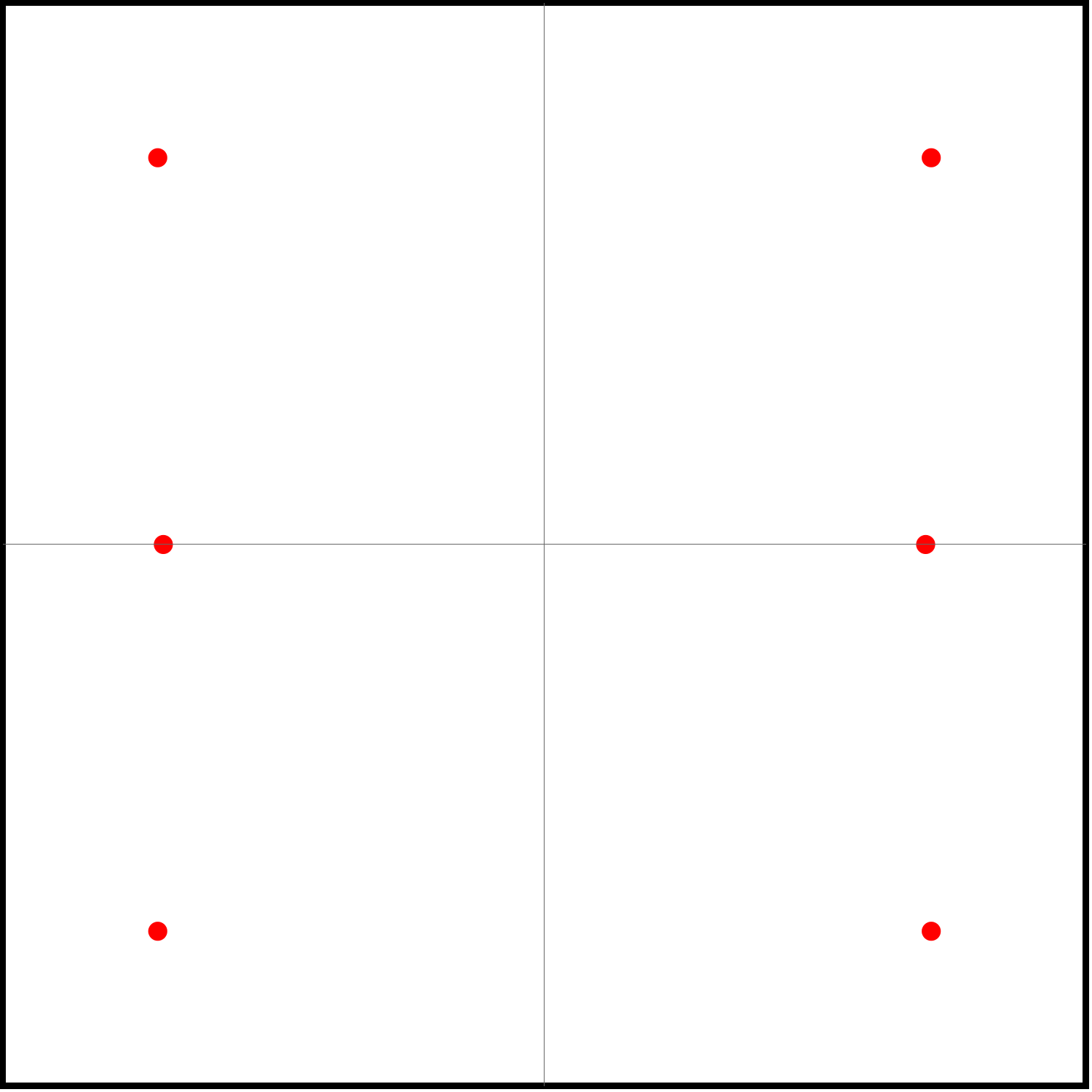} &
\includegraphics[width=0.11\textwidth]{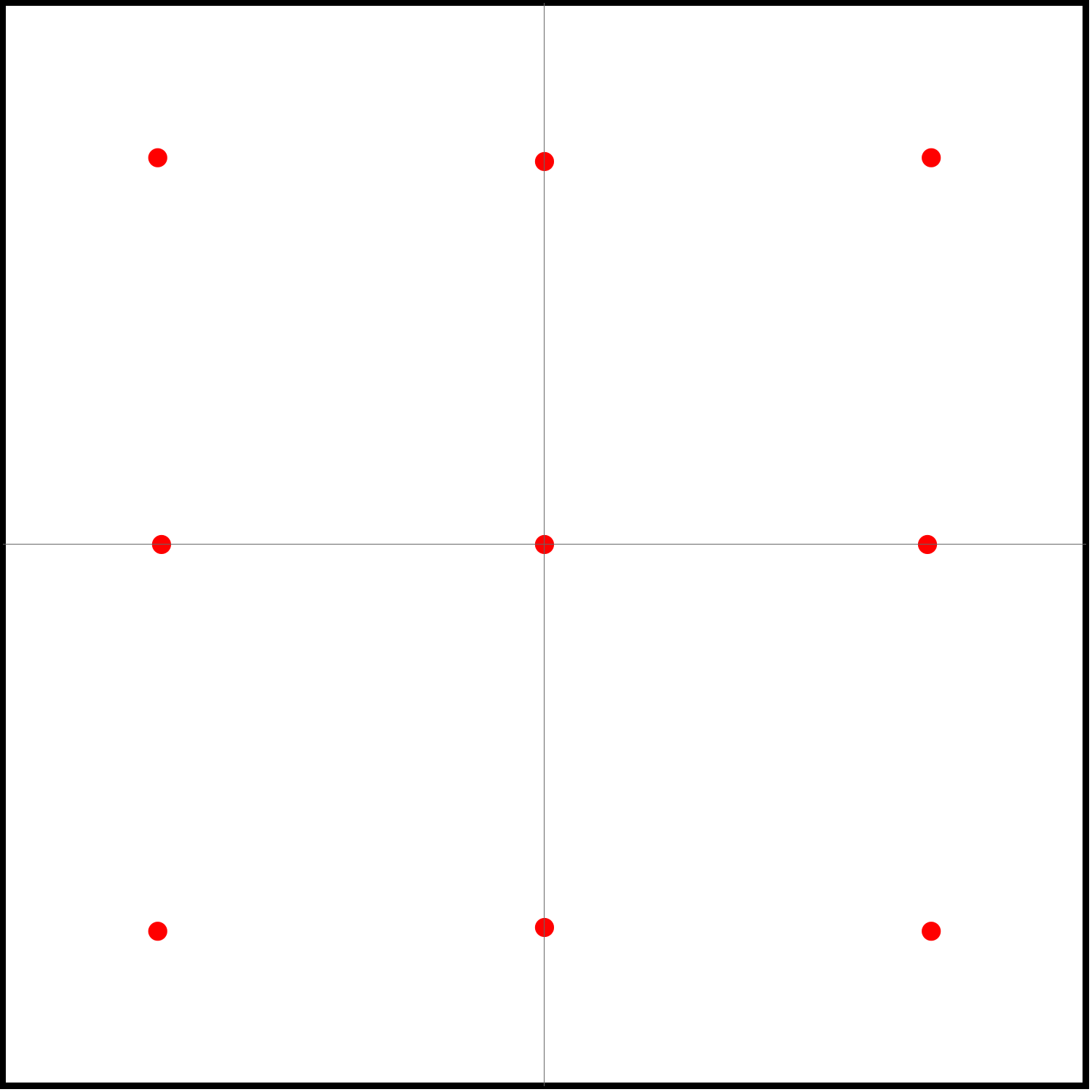} &
\includegraphics[width=0.11\textwidth]{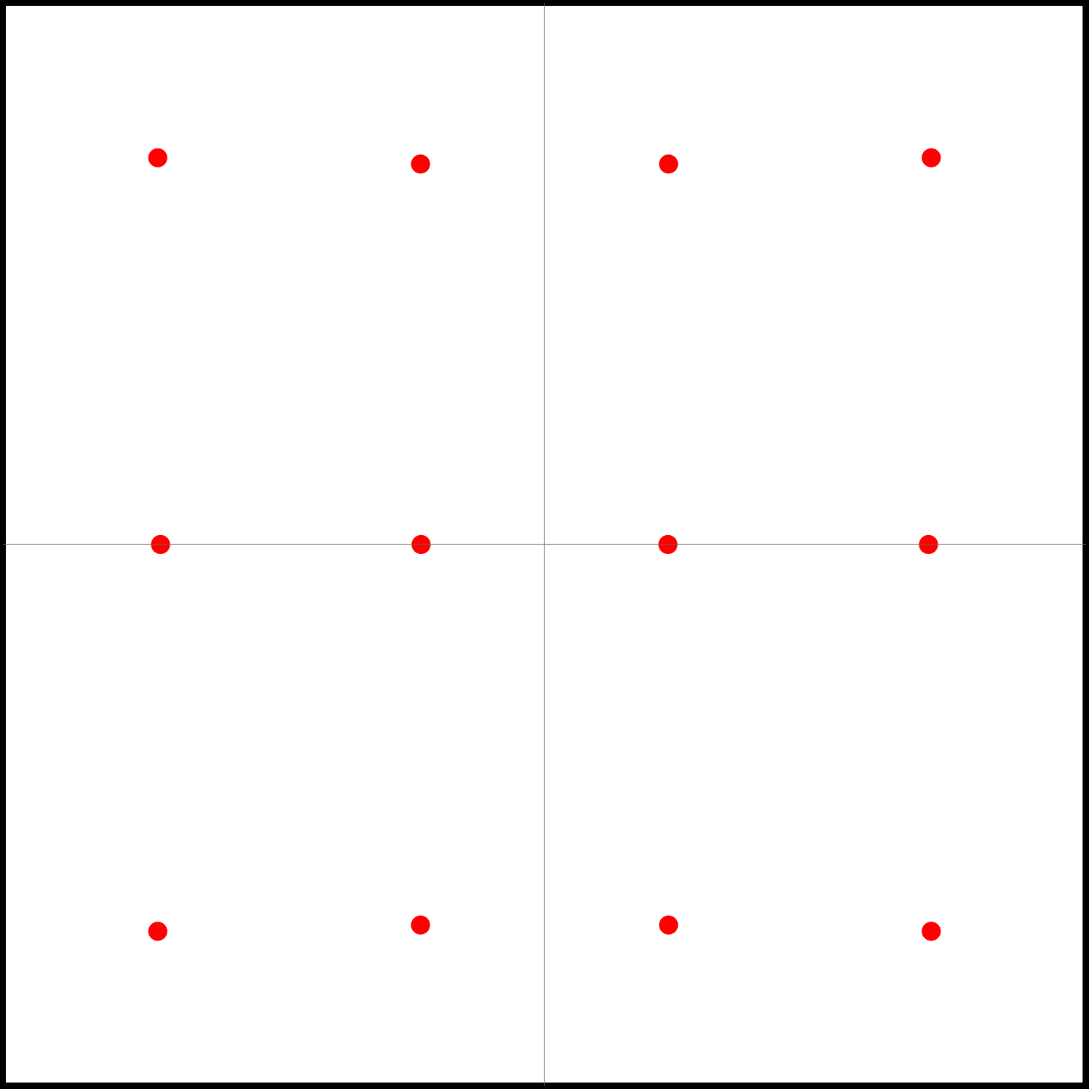} &
\includegraphics[width=0.11\textwidth]{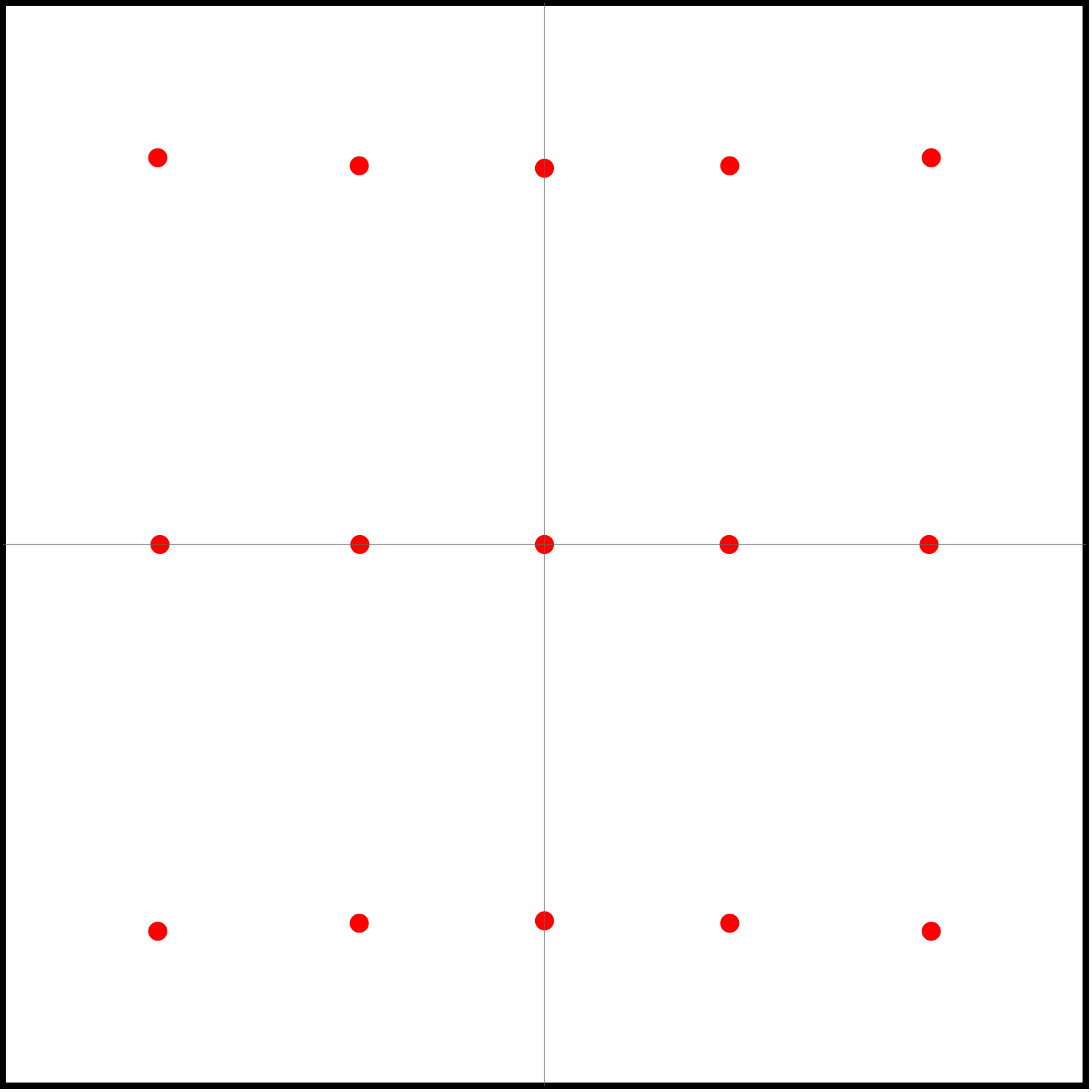} 
\\
\end{tabular}

\caption{Roots of the generalized Hermite polynomials
$H_{d,n}(z)$ with $1\le d\le 3$ and $2\le n\le 5$.
The horizontal and vertical plotting ranges in each panel are
automatically adjusted according to the locations of the roots.}

\label{fig:zeros of the generalized Hermite polynomials}
\end{figure}

\subsection{Generalized Okamoto polynomials}
\label{subsec:generalized-Okamoto-polynomials}

The generalized Okamoto polynomials associated with the \(-2z/3\) hierarchy of rational solutions of \(\mathrm{P}_{\mathrm{IV}}\) can be characterized by nonlinear recurrence relations \cite{NoumiYamada1999,clarkson2003fourth}. Determinant representations for these polynomials were obtained later by Kajiwara and Ohta \cite{kajiwara1998determinant}, while Noumi and Yamada \cite{NoumiYamada1999} derived an equivalent formulation in terms of \(3\)-reduced Schur functions.

For two nonnegative integers \(N_1\) and \(N_2\), let \(N=N_1+N_2\). When \(N\geq1\), take \(\boldsymbol{\mu}=(\mu_1,\mu_2,\ldots,\mu_N)\), where
\[
\mu_j= \begin{cases} 3j-2, & 1\leq j\leq N_1,\\[2mm] 3(j-N_1)-1, & N_1+1\leq j\leq N. \end{cases}
\]
We set \(Q_{0,0}(z)=1\). For \(N\geq1\), the generalized Okamoto polynomial $Q_{N_1,N_2}(z)=P_{\boldsymbol{\mu}}(z)$ can be written as
\begin{equation}
\label{eq:generalized-Okamoto-determinant-expanded}
\begin{aligned} Q_{N_1,N_2}(z) &= C_{\boldsymbol{\mu}}
\begin{vmatrix}
p_1(z) & p_0(z) & \cdots & p_{2-N}(z)\\
p_4(z) & p_3(z) & \cdots & p_{5-N}(z)\\
\vdots & \vdots & \ddots & \vdots\\
p_{3N_1-2}(z) & p_{3N_1-3}(z) & \cdots & p_{2N_1-N_2-1}(z)\\
p_2(z) & p_1(z) & \cdots & p_{3-N}(z)\\
p_5(z) & p_4(z) & \cdots & p_{6-N}(z)\\
\vdots & \vdots & \ddots & \vdots\\
p_{3N_2-1}(z) & p_{3N_2-2}(z) & \cdots & p_{2N_2-N_1}(z)
\end{vmatrix}\\
&=
C_{\boldsymbol{\mu}}
\operatorname{Wr}\left[
p_1(z),p_4(z),\ldots,p_{3N_1-2}(z);
p_2(z),p_5(z),\ldots,p_{3N_2-1}(z)
\right].
\end{aligned}
\end{equation}
The first and second blocks of rows in \eqref{eq:generalized-Okamoto-determinant-expanded} are omitted when \(N_1=0\) and \(N_2=0\), respectively. The normalization factor \(C_{\boldsymbol{\mu}}\) is given by \eqref{eq:Wronskian-normalization} and makes \(Q_{N_1,N_2}(z)\) monic.
By \eqref{eq:Pmu-degree},
\begin{equation}
\label{eq:degree-generalized-Okamoto}
d_{N_1,N_2} := \deg Q_{N_1,N_2} = N_1^2-N_1N_2+N_2^2+N_2.
\end{equation}
Since the generating function \eqref{eq:generating-function-pk} implies
\begin{equation}
p_k(-z)=(-1)^k p_k(z),
\end{equation}
the polynomial gives the parity relation
\begin{equation}
\label{eq:parity-generalized-Okamoto}
Q_{N_1,N_2}(-z) = (-1)^{d_{N_1,N_2}}Q_{N_1,N_2}(z).
\end{equation}
Moreover, \(Q_{N_1,N_2}(z)\) has real coefficients. Combining this fact with the parity relation \eqref{eq:parity-generalized-Okamoto}, we find that whenever \(z_0\) is a root, so are \(-z_0\), \(z_0^*\), and \(-z_0^*\).   It has been proved that all roots of the generalized Okamoto polynomials are simple \cite{NoumiYamada1999,roffelsen2025real}. Consequently, the origin is a root if and only if \(d_{N_1,N_2}\) is odd. Their root patterns were investigated by Clarkson \cite{clarkson2003fourth}.  Roffelsen and Stokes later derived exact formulas for the numbers of their real and purely imaginary roots. They also established several interlacing properties \cite{roffelsen2025real}.

Figure~\ref{fig:zeros-generalized-Okamoto-polynomials} illustrates the root configurations of \(Q_{N_1,N_2}(z)\), whose relation to the rogue wave patterns will be discussed in subsequent sections.

\begin{figure}[H]
\centering

\begingroup
\setlength{\tabcolsep}{2pt}
\renewcommand{\arraystretch}{1.02}

\begin{tabular}{c cccc}
&
$N_1=1$ &
$N_1=2$ &
$N_1=3$ &
$N_1=4$
\\


\rotatebox{90}{\shortstack{{$\quad N_2=1$}}} &
\includegraphics[width=0.11\textwidth]
{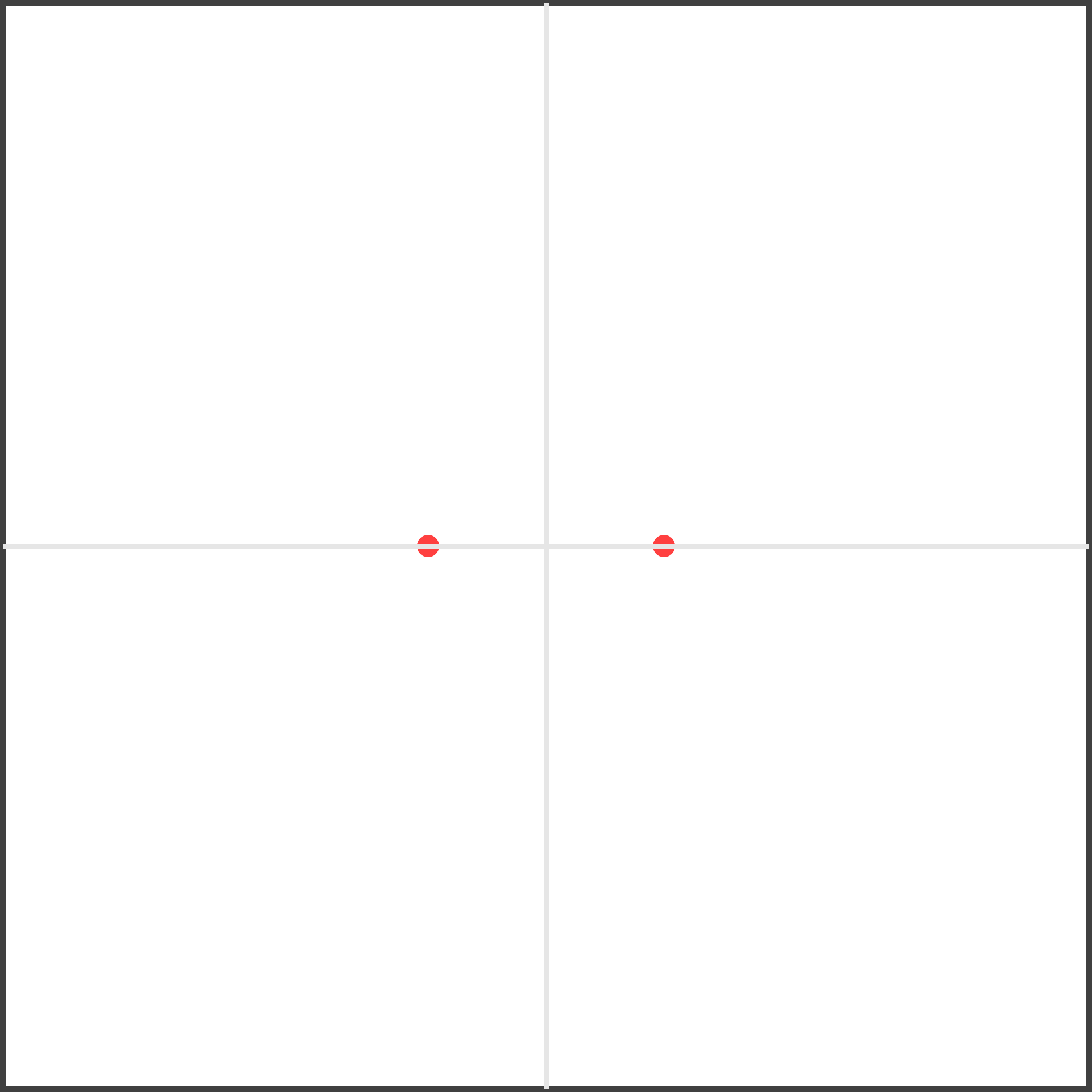} &
\includegraphics[width=0.11\textwidth]
{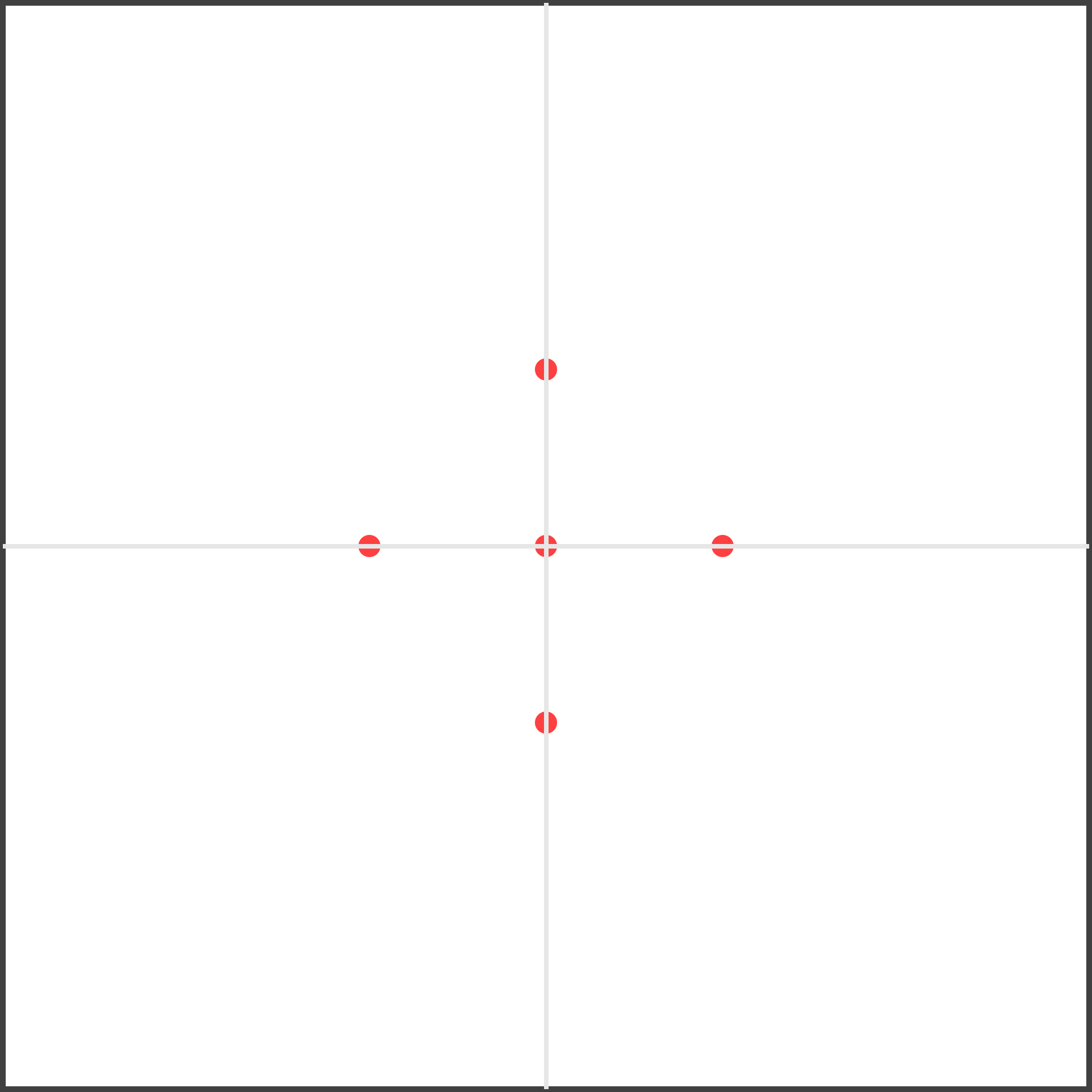} &
\includegraphics[width=0.11\textwidth]
{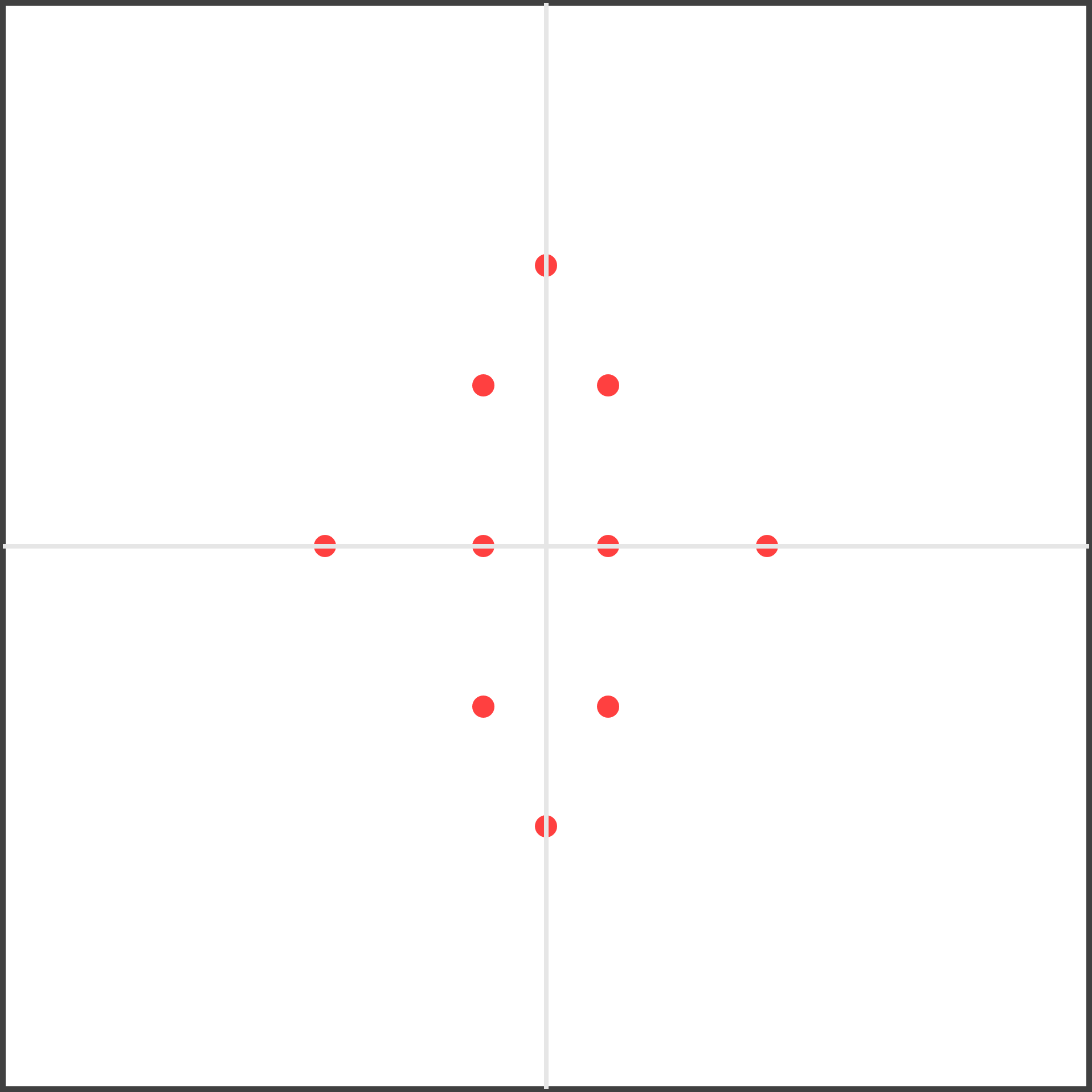} &
\includegraphics[width=0.11\textwidth]
{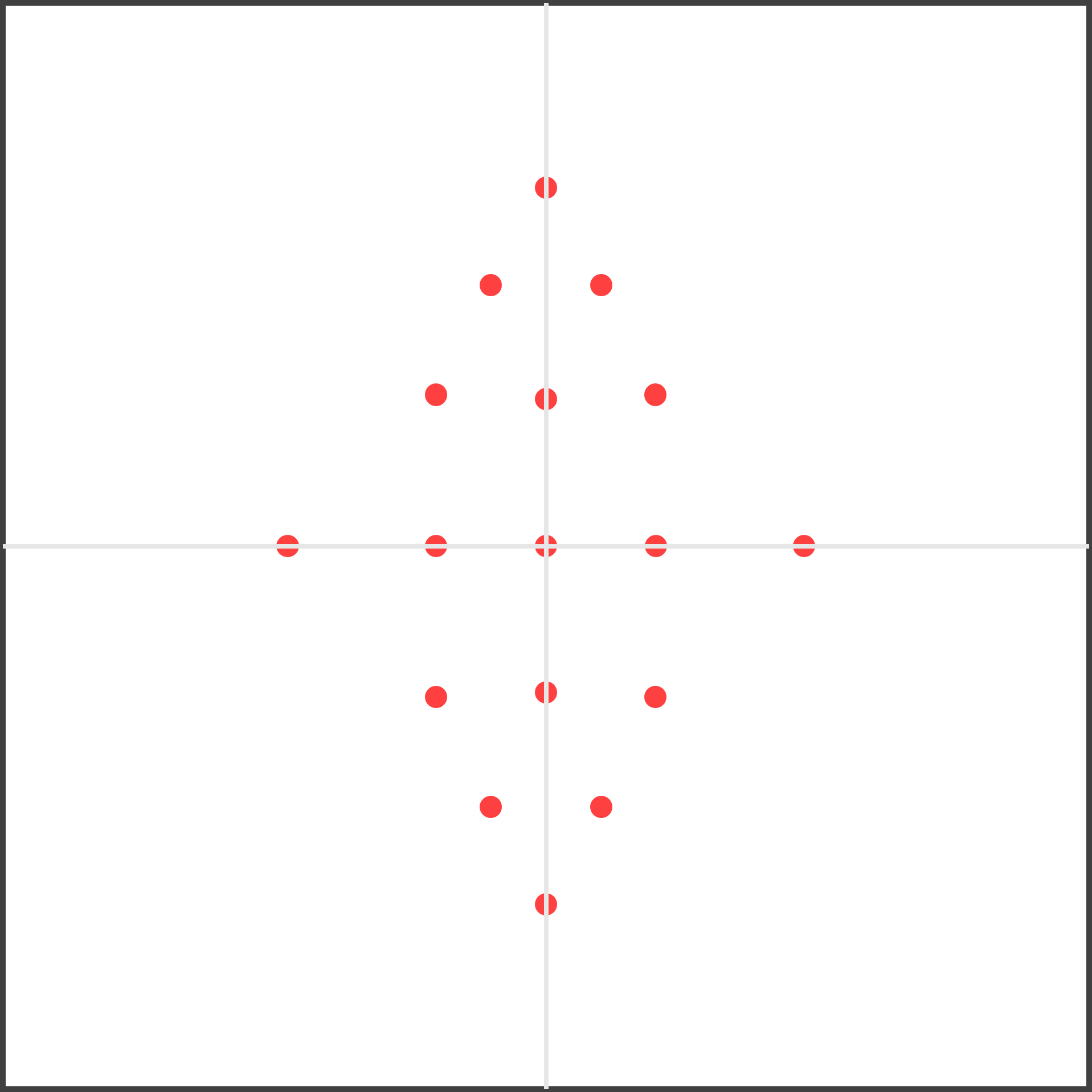}
\\

\rotatebox{90}{\shortstack{{$\quad N_2=2$}}} &
\includegraphics[width=0.11\textwidth]
{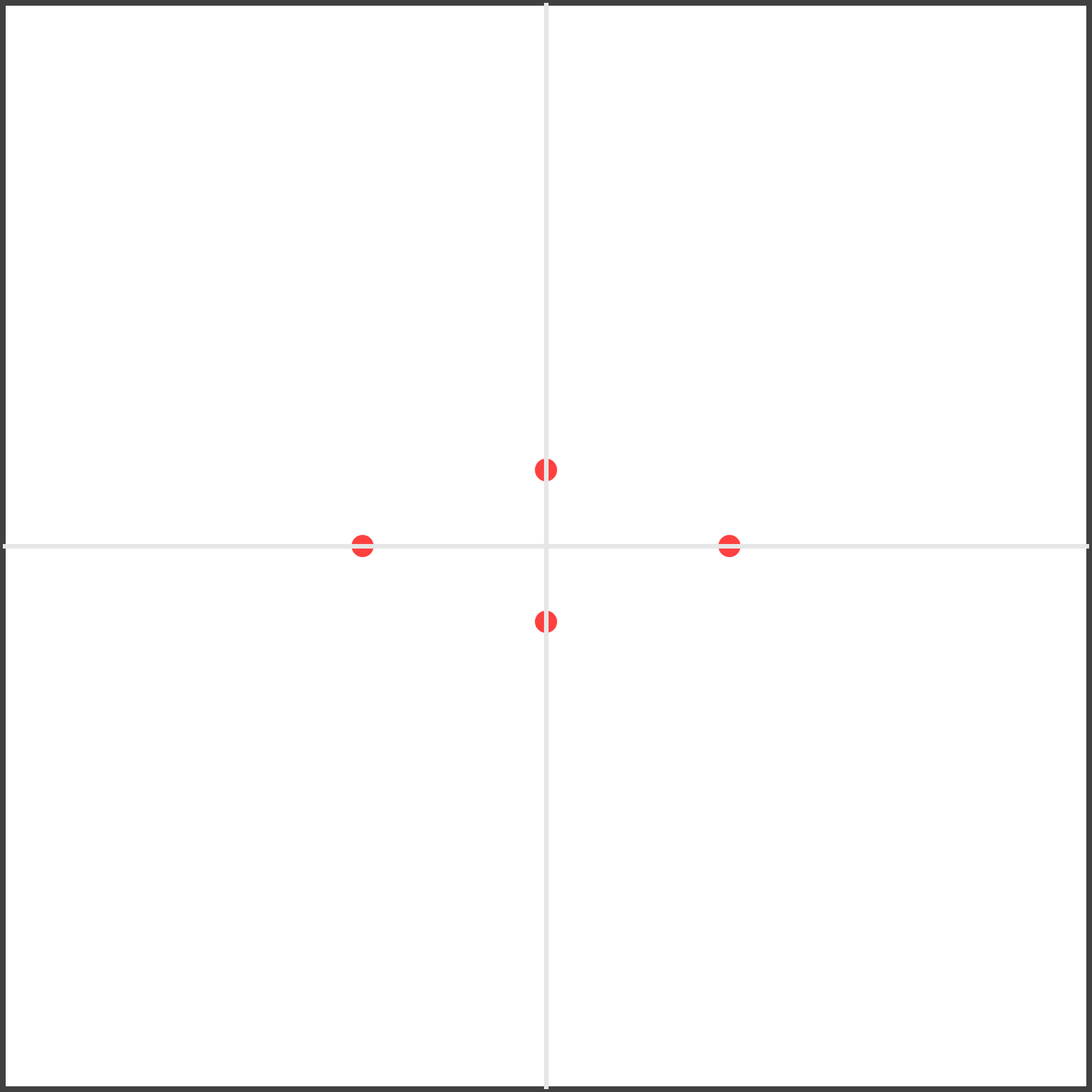} &
\includegraphics[width=0.11\textwidth]
{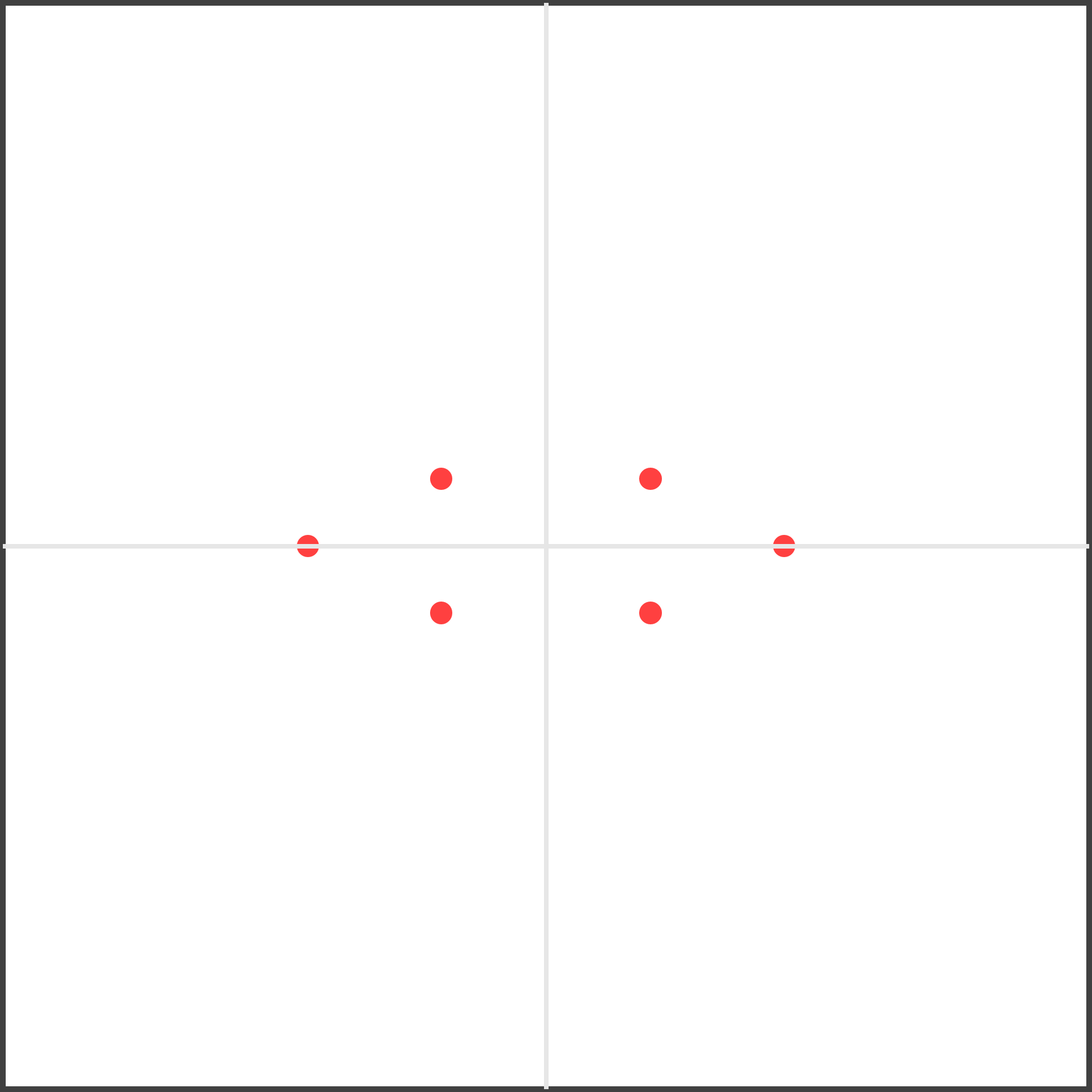} &
\includegraphics[width=0.11\textwidth]
{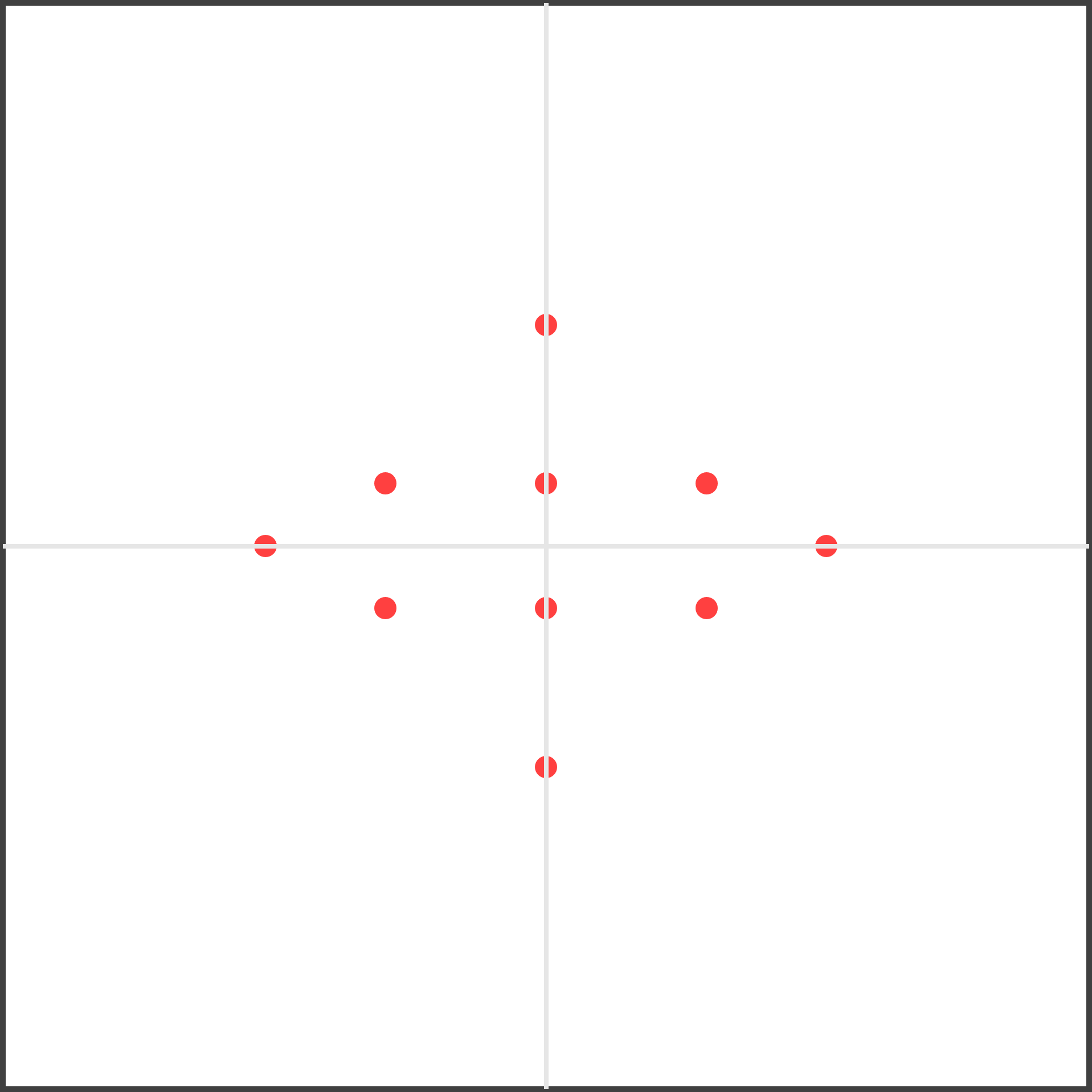} &
\includegraphics[width=0.11\textwidth]
{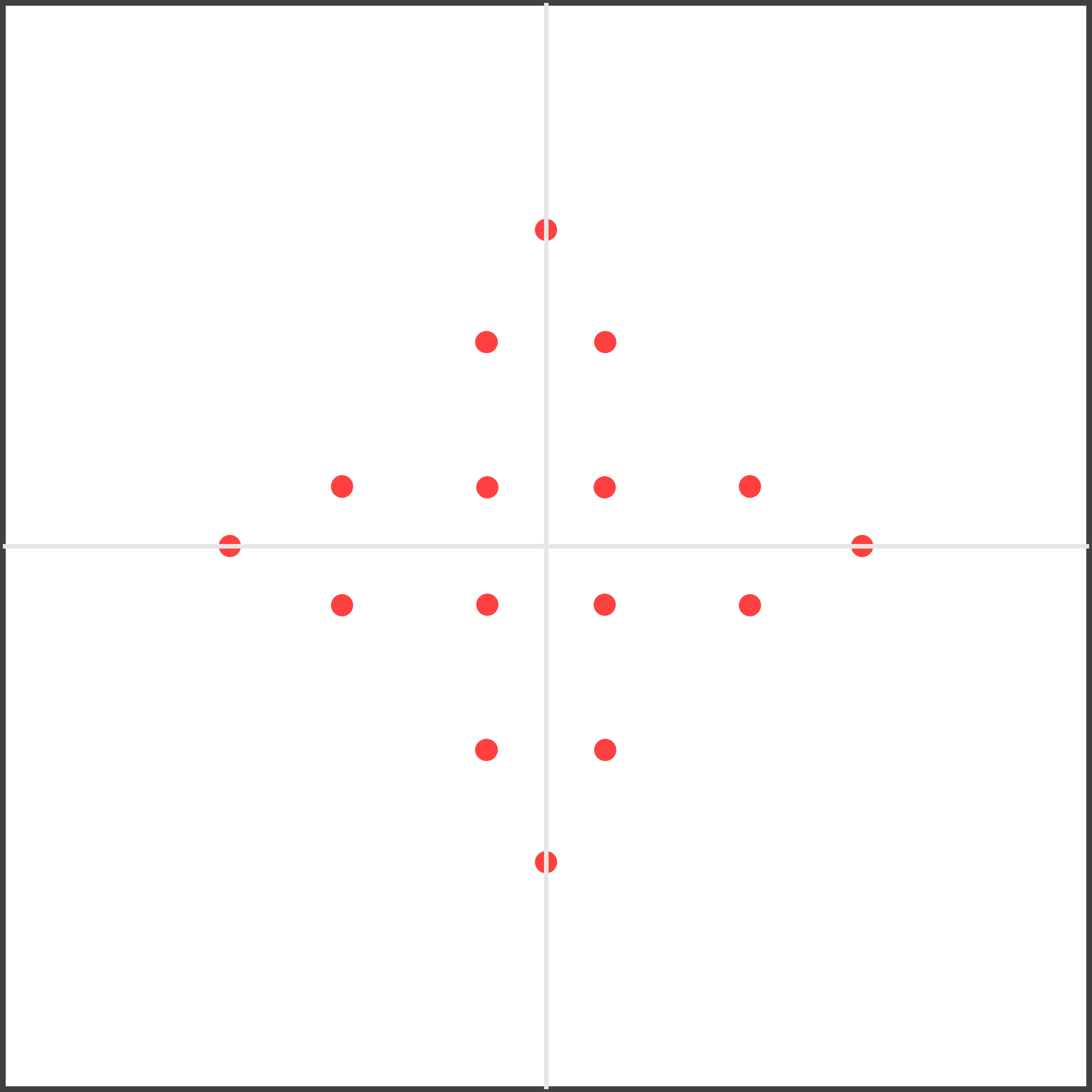}
\\

\rotatebox{90}{\shortstack{{$\quad N_2=3$}}} &
\includegraphics[width=0.11\textwidth]
{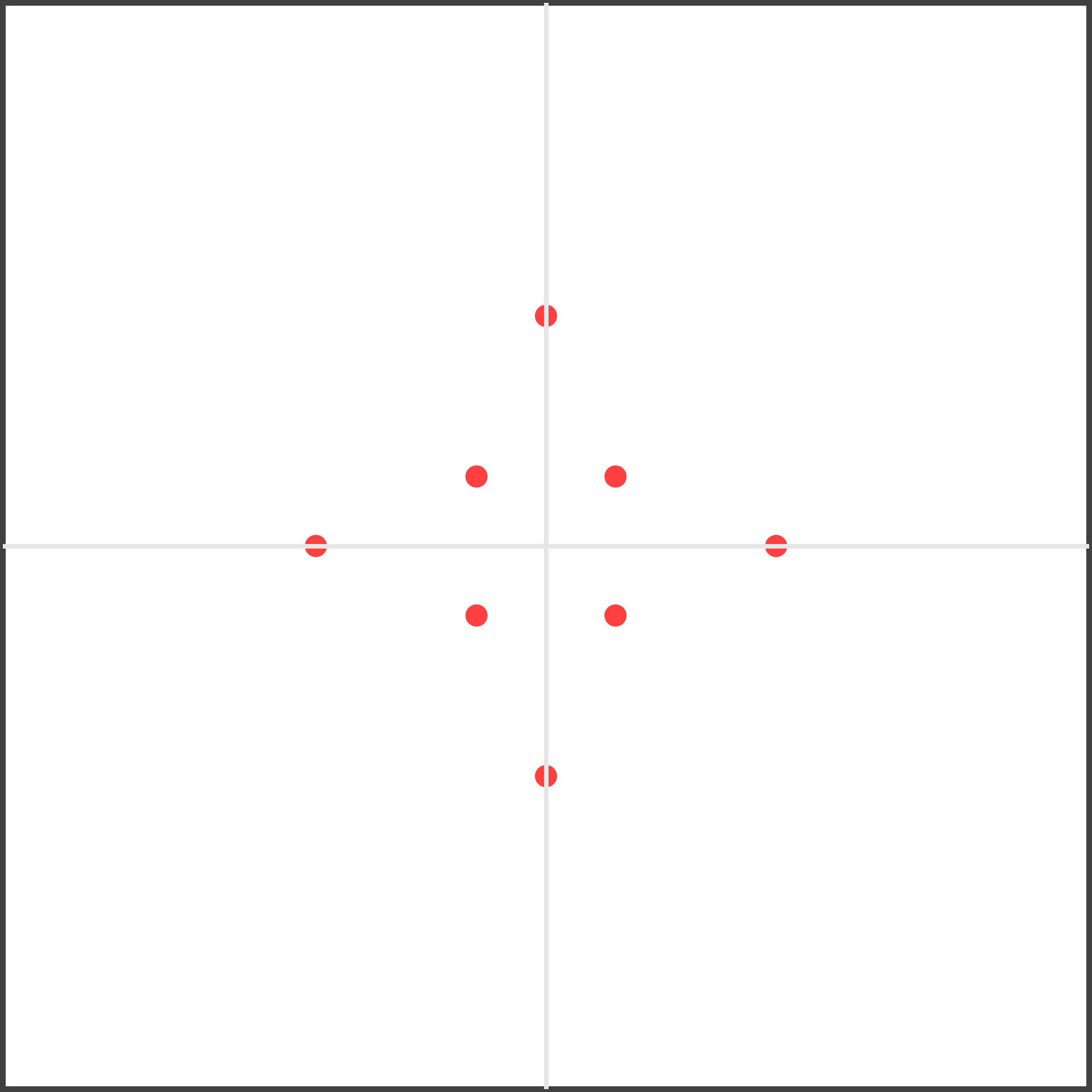} &
\includegraphics[width=0.11\textwidth]
{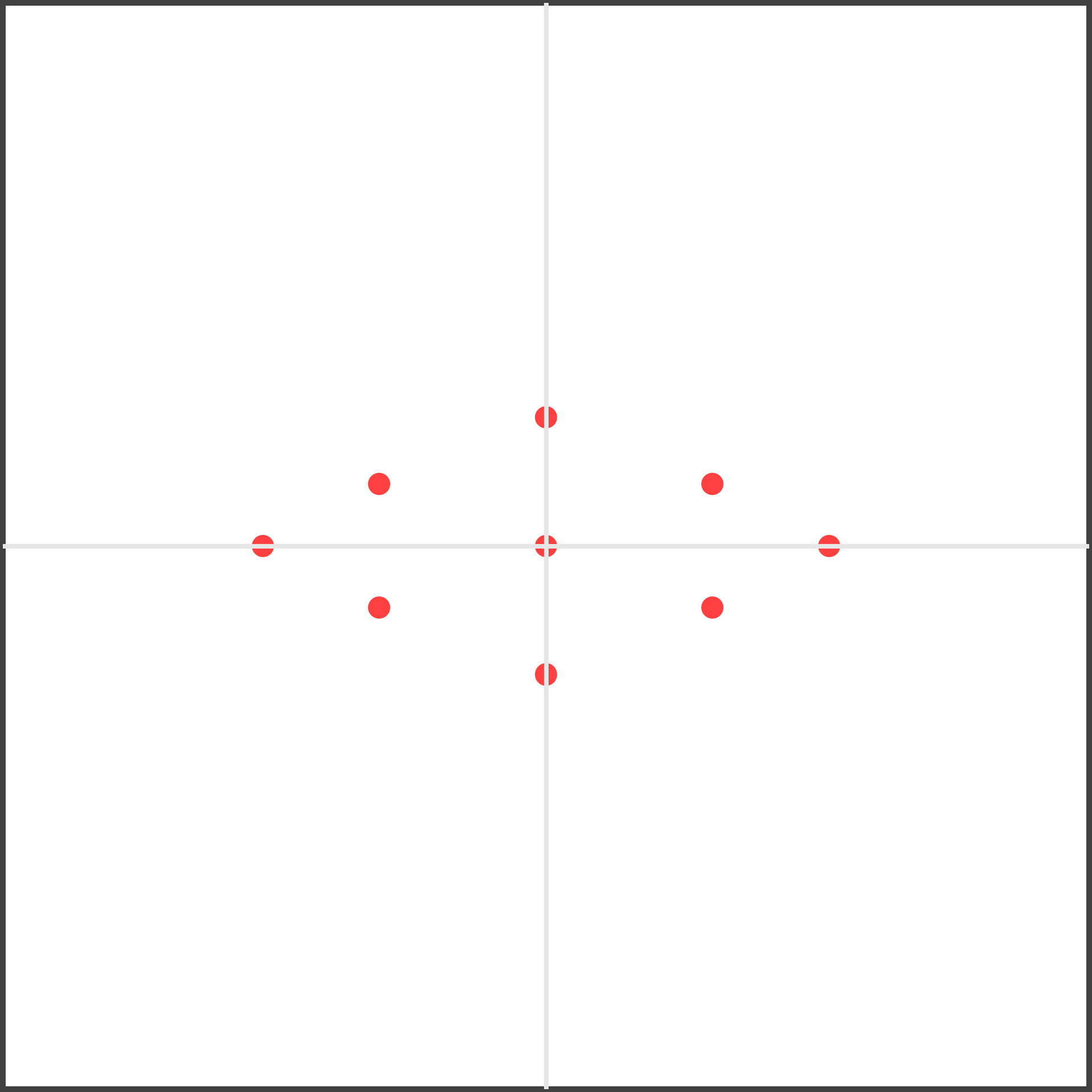} &
\includegraphics[width=0.11\textwidth]
{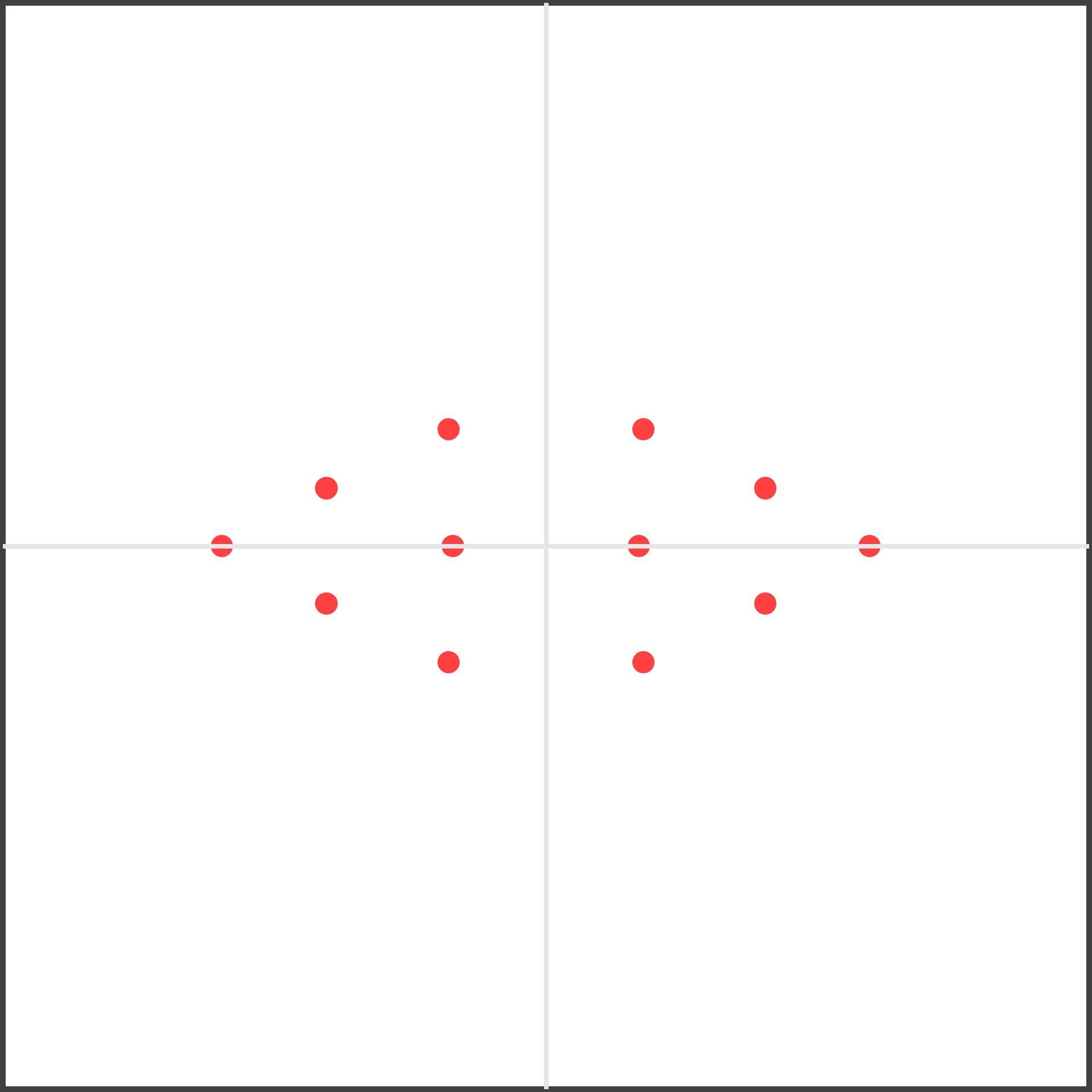} &
\includegraphics[width=0.11\textwidth]
{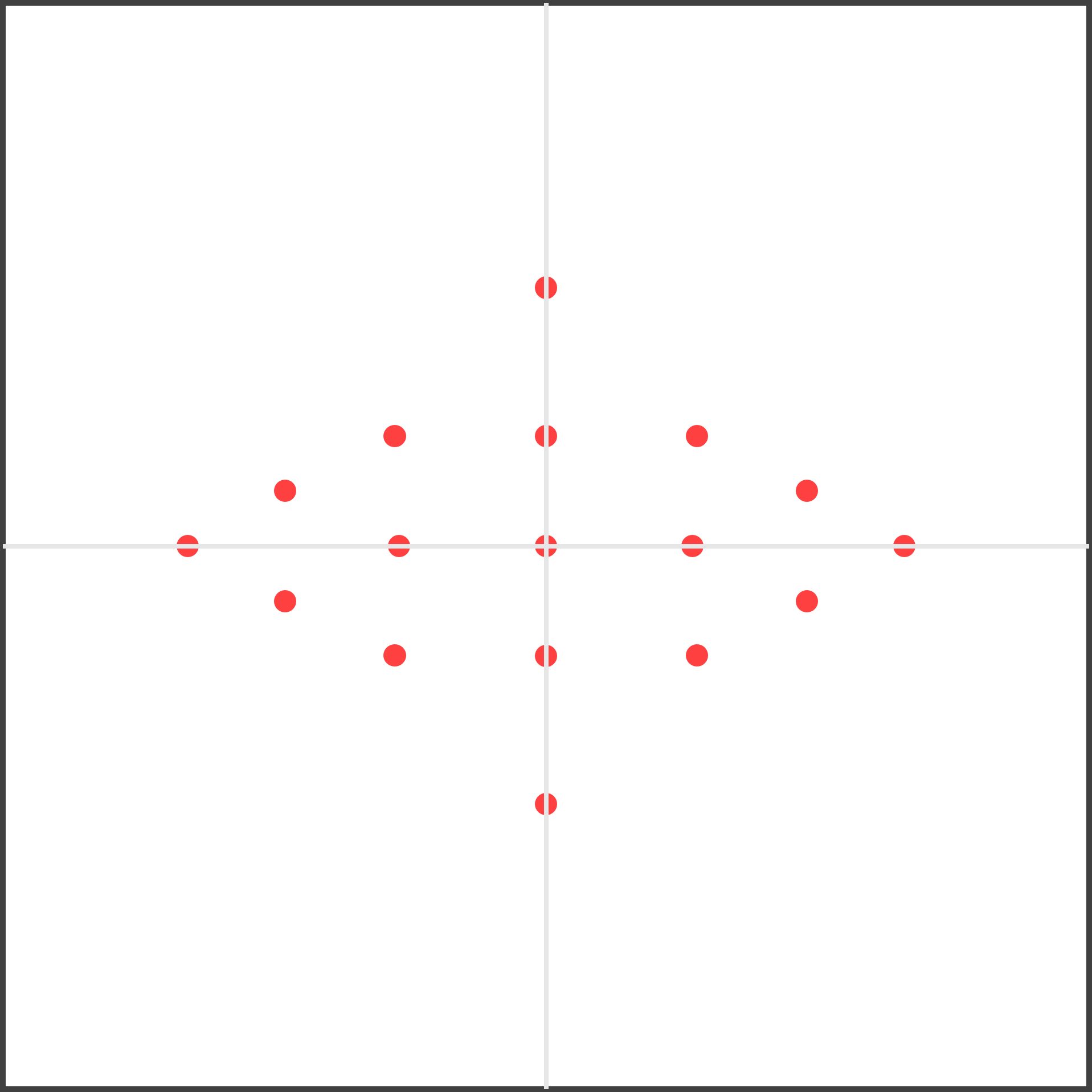}
\\

\rotatebox{90}{\shortstack{{$\quad N_2=4$}}} &
\includegraphics[width=0.11\textwidth]
{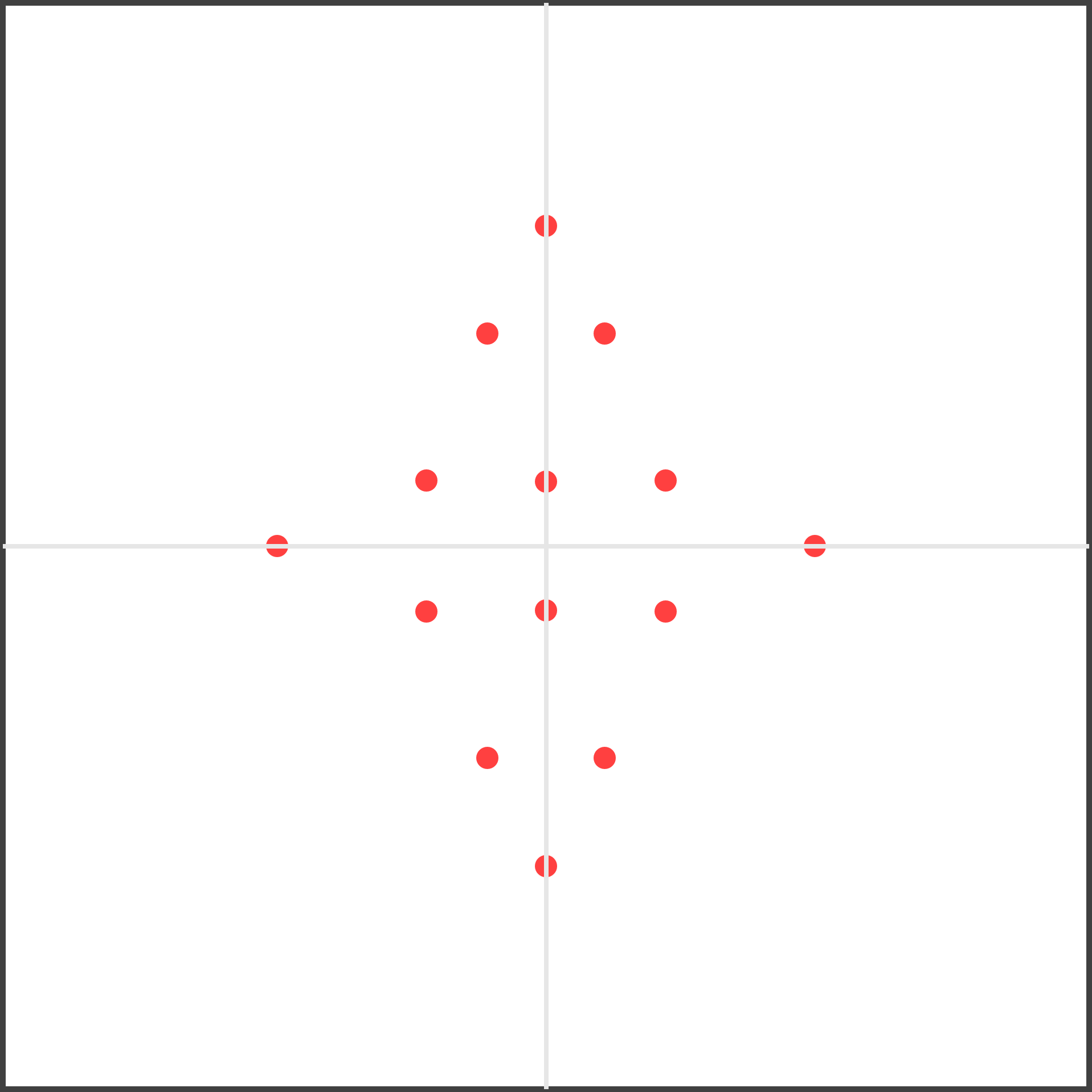} &
\includegraphics[width=0.11\textwidth]
{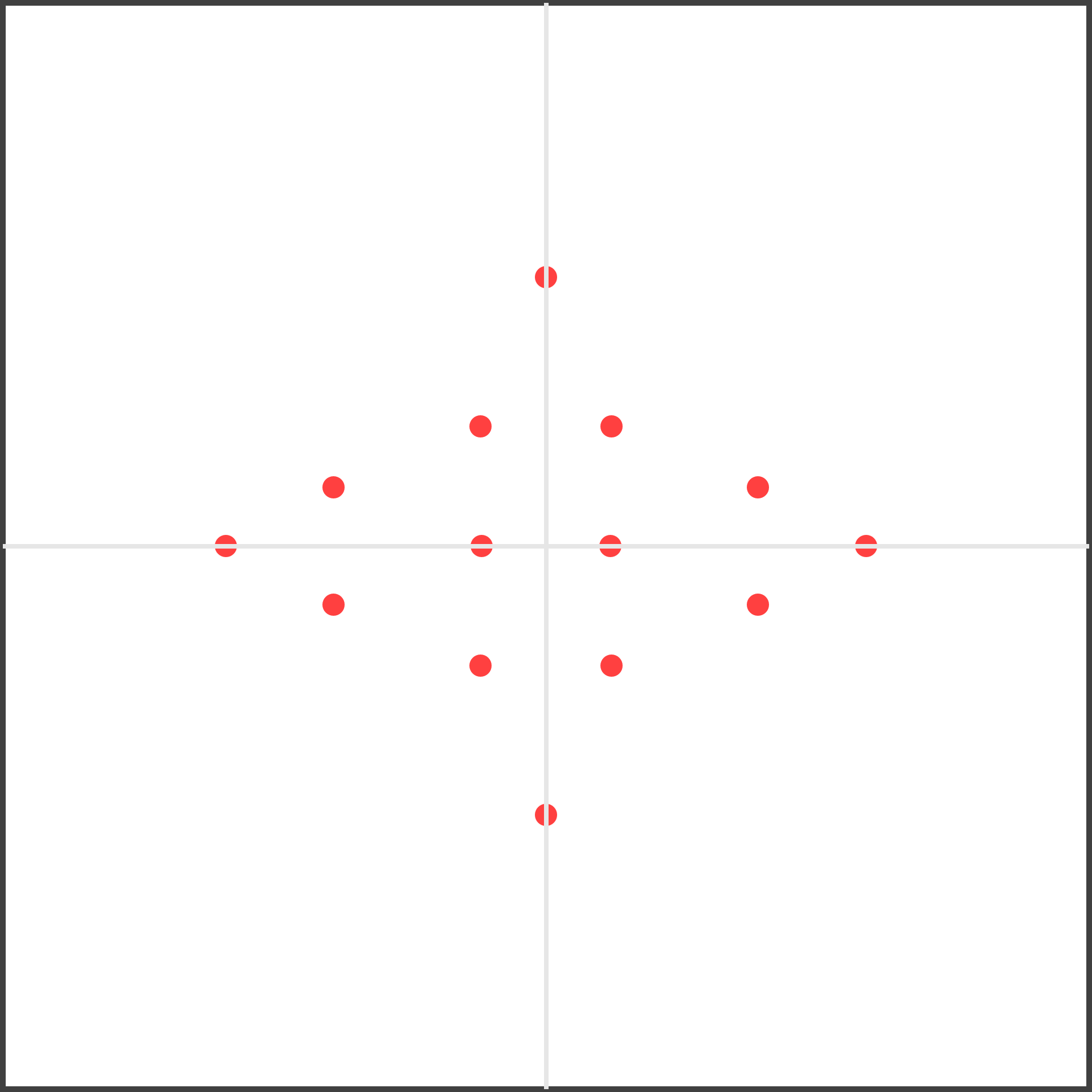} &
\includegraphics[width=0.11\textwidth]
{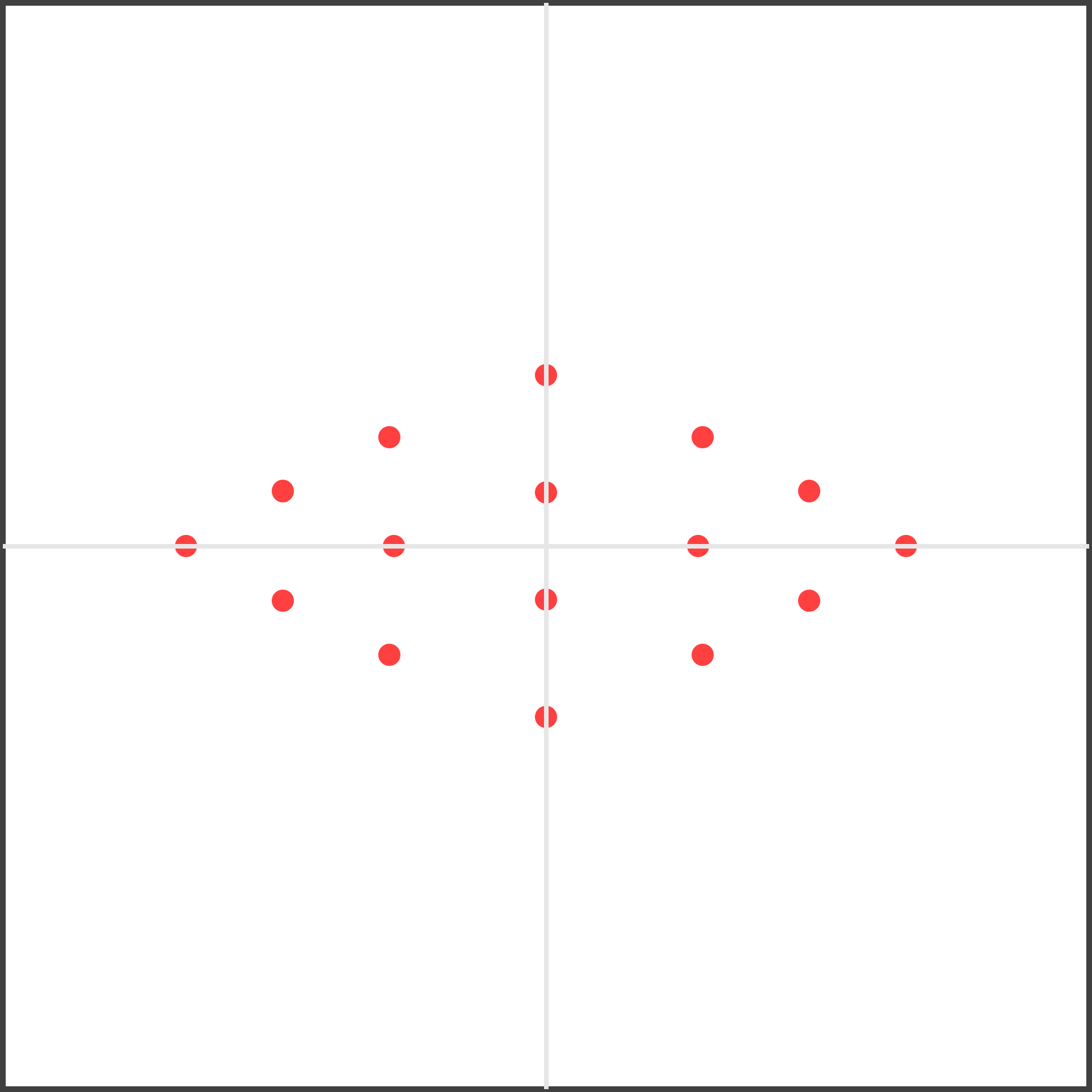} &
\includegraphics[width=0.11\textwidth]
{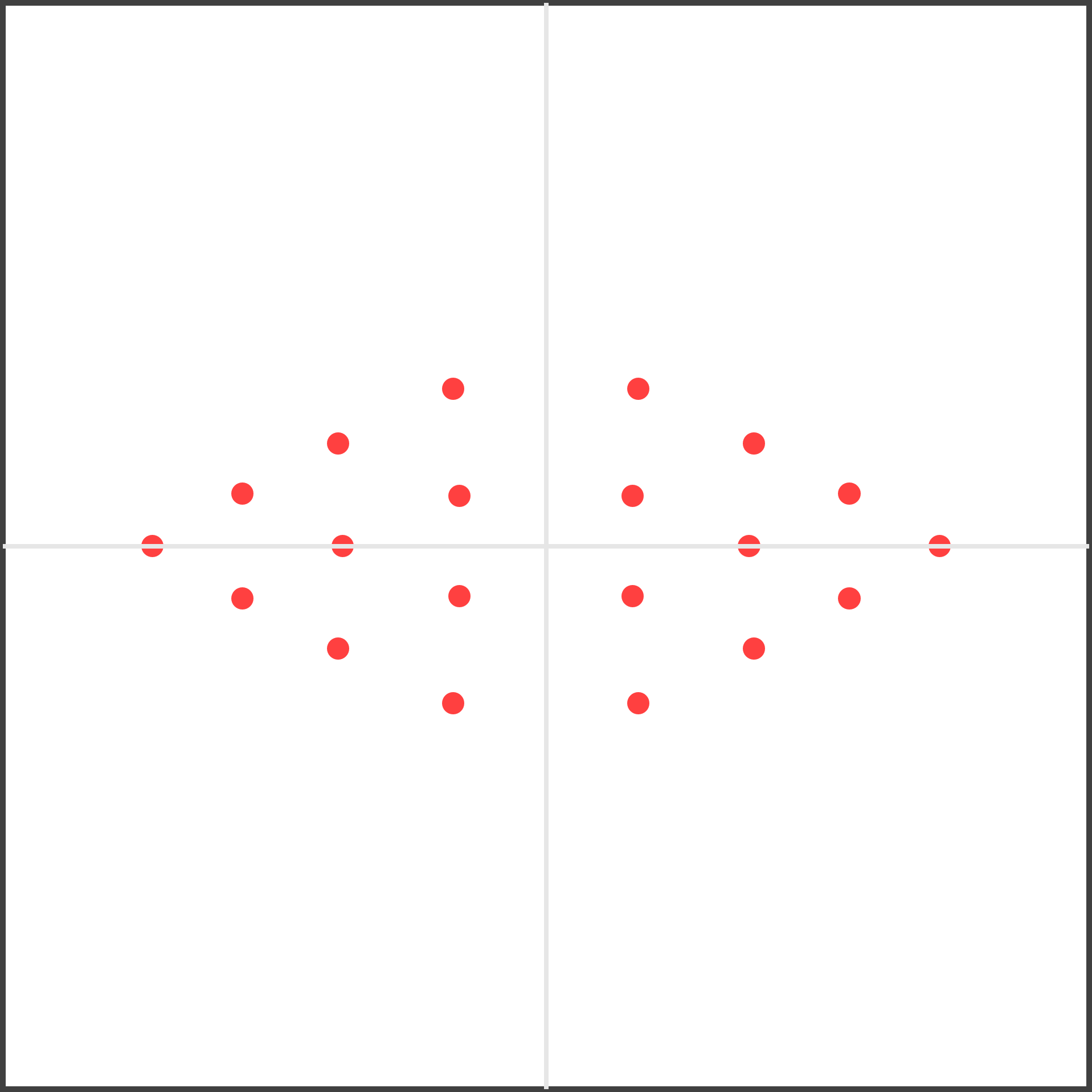}
\\


\end{tabular}

\endgroup


\caption{Roots of the generalized Okamoto polynomials
$Q_{N_1,N_2}(z)$ for $1\leq N_1,N_2\leq4$.
The horizontal and vertical plotting ranges in each panel are
automatically adjusted according to the locations of the roots.}
\label{fig:zeros-generalized-Okamoto-polynomials}
\end{figure}

\newpage
\section{Preliminaries}

In this section, we present the rogue wave solutions for the multi-component NLS and Hirota equations. We first introduce some notation required for the rogue wave solutions and their patterns. 
The Schur polynomials $S_n(\mathbf{x})$ are defined by
\begin{equation} \sum_{n=0}^{\infty}S_n(\mathbf{x})\lambda^n=\exp\left(\sum_{k=1}^{\infty}x_k\lambda^k\right),
\end{equation}
where $\mathbf{x}=(x_1,x_2,\cdots)$. To be more specific, we have
\begin{equation}\label{Schur polynomials}
  S_{0}(\mathbf{x})=1, \quad S_{1}(\mathbf{x})=x_{1}, \quad S_{2}(\mathbf{x})=\frac{1}{2} x_{1}^{2}+x_{2},  \ldots, \quad S_{j}(\mathbf{x})=\sum_{l_{1}+2 l_{2}+\cdots+m l_{m}=j}\left(\prod_{i=1}^{m} \frac{x_{i}^{l_{i}}}{l_{i} !}\right).
\end{equation}
Further, we define $S_j(\mathbf{x})\equiv0$ for $ j<0 $.

\begin{lemma}[\cite{zhang2022rogue}]\label{lemma:multiple roots}
Let $M$ be a positive integer and $\lambda_1>0, r_j \not = 0,k_j$ be real constants, with $k_i\neq k_j$ when $i\neq j$, $i,j=1,2,\dots, M$.  Let $\mathcal{R}_M(z)$ be a rational function defined by
\begin{equation}\label{dimension reduction VS rat}
\mathcal{R}_M(z)= \sum_{j=1}^M\frac{r_j }{(z+  k_j)^2} + 1.
\end{equation}
Then $\mathcal{R}_M(z)=0$ has a pair of complex conjugate roots with nonzero imaginary parts of multiplicity $M$
\begin{equation}
     \lambda_1  \cos\left[  \frac{\pi}{M+1}\right] -  k_1  \pm  \mathrm{i}  \lambda_1 \sin\left[  \frac{\pi}{M+1}\right],
\end{equation}
if the parameters $k_j$, $j=2,\ldots,M$, and $r_j$,
$j=1,\ldots,M$, satisfy
\begin{eqnarray}
  k_j  &=& k_1 +  \lambda_1 \left( \sin\left[\frac{\pi}{M+1}\right]  \cot\left[\frac{j\pi}{M+1}\right] -  \cos\left[\frac{\pi}{M+1}\right]\right),
\end{eqnarray}
and
\begin{eqnarray}
  r_j
  = (-1)^{j+1}  \prod_{\substack{i=1 \\ i \not=j}}^M (k_j-k_i)^{-1}  \left(\lambda_1 \frac{\sin\left[  \frac{\pi}{M+1}\right] }{\sin\left[\frac{j\pi}{M+1}\right]} \right)^{M+1}.
\end{eqnarray}
\end{lemma}

\subsection{Rogue wave solutions in the multi-component NLS equation}
In this subsection, we recall the rogue wave solutions of the multi-component NLS equation. We first define the function $\mathcal{G}_M(p)$ by
\begin{eqnarray}
\mathcal{G}_M(p) &=& \sum_{j=1}^M\frac{\sigma_j \rho_j^2}{p-\mathrm{i} k_j} +2 p, \label{definition of G(p))} 
\end{eqnarray}
where $\sigma_j = \pm 1$ and $\rho_j >0, k_j$ are real constants,  $j=1,2,\cdots,M$. We may deduce from Lemma \ref{lemma:multiple roots} that if   $\sigma_j, \rho_j $ and $ k_j, j=1,2,\cdots,M$, satisfy the constraints
\begin{equation}\label{contraints of multiple zeros}
\begin{aligned}
 k_j  &= k_1 +  \lambda_1 \left( \sin[  \pi /(M+1)]  \cot[j \pi /(M+1)] -  \cos[ \pi /(M+1)]\right),\\
 \sigma_j\rho_j^2
  &= 2 (-1)^{j+1}  \prod_{\substack{i=1 \\ i \not=j}}^M (k_j-k_i)^{-1}  \left(\lambda_1 \frac{\sin[  \pi /(M+1)] }{\sin[j \pi /(M+1)]} \right)^{M+1},
  \end{aligned}
\end{equation}
then the algebraic equation
\begin{equation}
\mathcal{G}_M^{\prime}(p)=0
\end{equation}
has a pair of non-imaginary roots of multiplicity $M$ given by
\begin{equation}
 \label{eq:NLS-multiple-root}
      \pm  \lambda_1 \sin[ \pi /(M+1)]-\mathrm{i} \lambda_1  \cos[ \pi /(M+1)] + \mathrm{i} k_1 .
\end{equation}
This is the case that we will encounter later. Additionally, we will need the function $p(\kappa)$ defined by 
\begin{eqnarray}
   \mathcal{G}_M(p(\kappa)) &=& \dfrac{\mathcal{G}_M(p(0))}{M+1} \sum_{n=1}^{M+1} \exp\left(\exp\left(\dfrac{2 n \pi \mathrm{i}}{M+1}\right) \kappa\right)  \nonumber
   \\
    &=& \dfrac{\mathcal{G}_M(p(0))}{M+1} \sum_{n=1}^{M+1} \exp\left({\cos{\left(\dfrac{2 n \pi }{M+1}\right)} \kappa}\right) \cos\left(\sin \left(\dfrac{2 n \pi }{M+1}\right)\kappa\right). \label{definition of p(kappa))}
\end{eqnarray}
\begin{theorem}
[\cite{zhang2022rogue}] \label{RW solutions of vector NLS}
 Let $M$ be a positive integer,  $\rho_j >0, k_j$ be real constants, and $\sigma_j=1 $, where $j=1,2,\cdots,M$. Assume  $ \rho_j $ and $ k_j$ are given by \eqref{contraints of multiple zeros}. Let $\mathcal{G}_M(p), p(\kappa)$ be functions defined by \eqref{definition of G(p))} and \eqref{definition of p(kappa))} respectively, and $p_0$ be a root of $\mathcal{G}_M^{\prime}(p)$ of multiplicity $M$.
Let $\boldsymbol{x}_{I}^{\pm}=\left(x_{1,I}^{\pm}, x_{2,I}^{\pm}, \cdots\right), I = 1, 2, \dots, M$, and $\boldsymbol{s}=\left(s_{1}, s_{2}, \cdots\right)$ be the vectors defined by
\begin{eqnarray}
&&x_{i,I}^{+}=\alpha_i x+\beta_i \mathrm{i} t+ \sum_{j=1}^{M}n_j \theta_{ij} + a_{i,I}, \label{values of x+} \\
&&x_{i,I}^{-}=\alpha_i^* x-\beta_i^* \mathrm{i} t-\sum_{j=1}^{M}n_j \theta_{ij}^* +a_{i,I}^*,\\
&&\ln \left[\frac{1}{\kappa}\left(\frac{p_{0}+p_{0}^{*}}{p_{1}}\right)\left(\frac{p(\kappa)-p_{0}}{p(\kappa)+p_{0}^{*}}\right)\right]=\sum_{r=1}^{\infty} s_{r} \kappa^{r}, 
\end{eqnarray}
where the asterisk `$*$' represents complex conjugation, $p_0=p(0),~p_1=p^{\prime}(0)$, the $a_{i,I}$'s are arbitrary constants, and $\alpha_i,~\beta_i,~\theta_{ij},~j = 1, 2, \dots, M$, are defined by the expansions
$$
p(\kappa)-p_{0}=\sum_{r=1}^{\infty} \alpha_{r} \kappa^{r}, \quad p^{2}(\kappa)-p_{0}^{2}=\sum_{r=1}^{\infty} \beta_{r} \kappa^{r},\quad\ln \frac{p(\kappa)-\mathrm{i} k_{j}}{p_{0}-\mathrm{i} k_{j}}=\sum_{r=1}^{\infty} \theta_{rj} \kappa^{r}.
$$
In this case, the $M$-component NLS equation \eqref{vector NLS} admits $\mathcal{N}$-th order rogue wave solutions
\begin{equation}
 u_{j,\mathcal{N}} = \rho_{j}\frac{g_{j,\mathcal{N}}}{f_{\mathcal{N}}}  e^{\mathrm{i}\left(k_j x+w_j t\right)} , \quad j=1,2,\cdots,M,
 \label{eq:NLS-solution}
\end{equation}
where 
\begin{equation} \label{definition of N-theorem}
w_j=\sum_{i=1}^{M}\sigma_i \rho_i^2 - k_j^2, \quad \mathcal{N} = \left(N_1, N_2, \ldots, N_M\right),
\end{equation}
with $N_j~(j=1,2,\cdots,M)$ being nonnegative integers,
and $f_{\mathcal N}$ and $g_{j,\mathcal N}$ are given by
\begin{equation}
f_{\mathcal{N}}=\tau_{\mathbf{n}_0}, \quad g_{j,\mathcal{N}}=\tau_{\mathbf{n}_j},
\end{equation}
with
$$
\mathbf{n}_0=(0,0, \ldots, 0) \in \mathbb{R}^M, \quad \mathbf{n}_j= \boldsymbol{e}_j,
$$
and $\boldsymbol{e}_j$ being the standard unit vector in $\mathbb{R}^M$. Here, $\tau_{\mathbf{n}}$ is given by the following $K \times K \,(1\leq K \leq M) $ block determinant
\begin{equation}
\label{tau-block matrix-theorem-1}
\tau_{\mathbf{n}}=\det\left(\begin{array}{llll}
\tau_{\mathbf{n}}^{[I_1,I_1]} & \tau_{\mathbf{n}}^{[I_1,I_2]}&\cdots&\tau_{\mathbf{n}}^{[I_1,I_K]} \\
\tau_{\mathbf{n}}^{[I_2,I_1]} & \tau_{\mathbf{n}}^{[I_2,I_2]}&\cdots&\tau_{\mathbf{n}}^{[I_2,I_K]} \\ \vdots & \vdots& \ddots &\vdots \\\tau_{\mathbf{n}}^{[I_K,I_1]} & \tau_{\mathbf{n}}^{[I_K,I_2]}&\cdots&\tau_{\mathbf{n}}^{[I_K,I_K]}
\end{array}\right)_{N \times N},
\end{equation}
where
\begin{eqnarray}
&&\mathbf{n}=\left(n_1, n_2, \ldots, n_M\right),
\\
&& 1\leq I_1<I_2<\cdots<I_K\leq M,
\\
&& \tau_{\mathbf{n}}^{[I, J]}=\left(m_{(M+1)(i-1)+I,(M+1) (j-1)+J}^{(\mathbf{n},I, J)}\right)_{1 \leq i \leq N_I, 1 \leq j \leq N_J}, \quad 1 \leq I, J \leq M,\label{tau-entry of block matrix-theorem-1}
\end{eqnarray}
$n_1, n_2, \ldots, n_M $ are integers,  $ I_j, N_{I_j} \, (j=1,2,\cdots,K)$  are positive integers with $N_{I_1}+N_{I_2}+\cdots+N_{I_K} = N, N \geq K$ and $N_l = 0$ for $l \in \{1,2,\cdots,M\}  \backslash \{I_1,I_2,\cdots,I_K\}$, and the corresponding matrix elements of \eqref{tau-entry of block matrix-theorem-1} are defined by
\begin{eqnarray} \label{eq:m-entry}
    m_{i, j}^{(\mathbf{n},I, J)}&=&\sum_{v=0}^{\min (i, j)}\left[\frac{\left|p_{1}\right|^{2}}{\left(p_{0}+p_{0}^{*}\right)^{2}}\right]^{v} S_{i-v}\left(\boldsymbol{x}_I^{+}(\mathbf{n})+v \boldsymbol{s}\right) S_{j-v}\left(\boldsymbol{x}_J^{-}(\mathbf{n})+v \boldsymbol{s}^{*}\right).
\end{eqnarray}
\end{theorem}

\subsection{Rogue wave solutions in the multi-component Hirota equation}

In this subsection, we recall the rogue wave solutions of the multi-component Hirota equation. First, we define the rational function $\mathcal{F}_M(p)$ associated with the multi-component Hirota equation as
\begin{equation}\label{eqn:definition of F(p)}
\mathcal{F}_M(p) = -\sum_{j=1}^M \frac{c_j \rho_j^2}{p-\mathrm{i} k_j} + p, 
\end{equation}
where $c_j = \pm 1$, and $\rho_j > 0, k_j$ are real constants for $j=1,2,\dots,M$. It follows from Lemma \ref{lemma:multiple roots} that if the parameters $\{c_j, \rho_j, k_j\}_{j=1}^M$ satisfy the following constraints:
\begin{equation}\label{eqn:constraints of multiple zeros}
\begin{aligned}
 k_j &= k_1 + \lambda_1 \left( \sin\left[\frac{\pi}{M+1}\right] \cot\left[\frac{j \pi}{M+1}\right] - \cos\left[\frac{\pi}{M+1}\right]\right), \\
 c_j \rho_j^2 &= (-1)^{j} \prod_{\substack{i=1 \\ i \not=j}}^M (k_j-k_i)^{-1}  \left(\lambda_1 \frac{\sin\left[  \frac{\pi}{M+1}\right] }{\sin\left[\frac{j\pi}{M+1}\right]} \right)^{M+1},
 \end{aligned}
\end{equation}
then $\mathcal{F}_M^{\prime}(p)=0$ possesses a pair of non-imaginary roots with multiplicity $M$, given by
\begin{equation} \label{root of multiplicity M}
 p = \pm \lambda_1 \sin\left[\frac{\pi}{M+1}\right] - \mathrm{i} \lambda_1 \cos\left[\frac{\pi}{M+1}\right] + \mathrm{i} k_1.
\end{equation}
Next, we define the auxiliary function $p(\kappa)$ by
\begin{equation}\label{eqn:definition of p(kappa)}
\begin{aligned}
 \mathcal{F}_M(p(\kappa)) 
= \frac{\mathcal{F}_M(p(0))}{M+1} \sum_{n=1}^{M+1} \exp\left[\cos\left(\frac{2 n \pi}{M+1}\right) \kappa\right] \cos\left[\sin \left(\frac{2 n \pi}{M+1}\right)\kappa\right],
\end{aligned}
\end{equation}
where $p(0)$ is given by \eqref{root of multiplicity M}.
Then, one can obtain the following rogue wave solutions and rational solutions of the multi-component Hirota equation \eqref{eqn: vector-Hirota}.

\begin{theorem}[\cite{wu2026zerosgeneralizedwronskianhermitepolynomials}]
\label{thm:rogue wave solutions}
Let $M$ be a positive integer,  $\rho_j >0, k_j$ be real constants, and $c_j=-1 $, where $j=1,2,\cdots,M$. Assume  $ \rho_j $ and $ k_j$ are given by \eqref{eqn:constraints of multiple zeros}. Let $\mathcal{F}_M(p), p(\kappa)$ be functions defined by \eqref{eqn:definition of F(p)} and \eqref{eqn:definition of p(kappa)} respectively, and $p_0$ be a root of $\mathcal{F}_M^{\prime}(p)$ of multiplicity $M$ with $\Re{p_0} \Im{p_0} \neq 0$. Then, the $M$-component Hirota equation \eqref{eqn: vector-Hirota} admits $\mathcal{N}$-th order rogue wave solutions
\begin{equation} \label{eq:Hirota-solution}
 v_{j,\mathcal{N}}(x,t) = \rho_{j}\frac{g_{j,\mathcal{N}}(x,t)}{f_{\mathcal{N}}(x,t)}  e^{\mathrm{i}\left(k_j x-w_j t\right)} , \quad j=1,2,\cdots,M,
\end{equation}
where $w_j=k_j^3+3k_j\sum_{i=1}^{M}c_i \rho_i^2+3\sum_{i=1}^{M}c_i\rho_i^2 k_i$, $\mathcal{N} = \left(N_1, N_2, \ldots, N_M\right)$ with $N_j~(j=1,2,\cdots,M)$ being nonnegative integers,
and $f_{\mathcal N}$ and $g_{j,\mathcal N}$ are given by
\begin{equation}
f_{\mathcal{N}}(x,t)=\tau_{\mathbf{n}_0}\left(x-3\sum_{i=1}^{M}c_i \rho_i^2 t,t\right), \quad g_{j,\mathcal{N}}(x,t)=\tau_{\mathbf{n}_j}\left(x-3\sum_{i=1}^{M}c_i \rho_i^2 t,t\right),
\end{equation}
with $\mathbf{n}_0=(0,0, \ldots, 0) \in \mathbb{R}^M$ and $\mathbf{n}_j= \mathbf{e}_j$ being the standard unit vector in $\mathbb{R}^M$. Here, $\tau_{\mathbf{n}}$ is given by the following $K \times K \, (1\leq K \leq M) $ block determinant
\begin{equation} \label{tau-block matrix-theorem-2}
\tau_{\mathbf{n}}=\det\left(\begin{array}{llll}
\tau_{\mathbf{n}}^{[I_1,I_1]} & \tau_{\mathbf{n}}^{[I_1,I_2]}&\cdots&\tau_{\mathbf{n}}^{[I_1,I_K]} \\
\tau_{\mathbf{n}}^{[I_2,I_1]} & \tau_{\mathbf{n}}^{[I_2,I_2]}&\cdots&\tau_{\mathbf{n}}^{[I_2,I_K]} \\ \vdots & \vdots& \ddots &\vdots \\\tau_{\mathbf{n}}^{[I_K,I_1]} & \tau_{\mathbf{n}}^{[I_K,I_2]}&\cdots&\tau_{\mathbf{n}}^{[I_K,I_K]}
\end{array}\right)_{N \times N},
\end{equation}
where
\begin{eqnarray}
&&\mathbf{n}=\left(n_1, n_2, \ldots, n_M\right),
\\
&& 1\leq I_1<I_2<\cdots<I_K\leq M,
\\
&& \tau_{\mathbf{n}}^{[I, J]}=\left(m_{(M+1)(i-1)+I,(M+1) (j-1)+J}^{(\mathbf{n},I, J)}\right)_{1 \leq i \leq N_I, 1 \leq j \leq N_J}, \quad 1 \leq I, J \leq M,  \label{tau-entry of block matrix-theorem-2}
\end{eqnarray}
$n_1, n_2, \ldots, n_M $ are integers,  $ I_j, N_{I_j} \, (j=1,2,\cdots,K)$  are positive integers with $N_{I_1}+N_{I_2}+\cdots+N_{I_K} = N, N \geq K$ and $N_l = 0$ for $l \in \{1,2,\cdots,M\}  \backslash \{I_1,I_2,\cdots,I_K\}$, and the corresponding matrix elements of \eqref{tau-entry of block matrix-theorem-2} are defined by
\begin{eqnarray} \label{eqn: definition of mij}
    m_{i, j}^{(\mathbf{n},I, J)}=\sum_{\nu=0}^{\min (i, j)}\left[\frac{\left|p_{1}\right|^{2}}{\left(p_{0}+p_{0}^{*}\right)^{2}}\right]^{\nu} S_{i-\nu}\left(\mathbf{x}_I^{+}(\mathbf{n})+\nu \mathbf{s}\right) S_{j-\nu}\left(\mathbf{x}_J^{-}(\mathbf{n})+\nu \mathbf{s}^{*}\right).
\end{eqnarray}
The vectors $\mathbf{x}_{I}^{\pm}
=\left(x_{1,I}^{\pm}, x_{2,I}^{\pm}, \cdots\right)$, $I=1,2,\dots,M$,
and $\mathbf{s}=\left(s_{1}, s_{2}, \cdots\right)$
are given by
\begin{eqnarray}
&&x_{i,I}^{+}=\alpha_i x+\beta_i t+ \sum_{j=1}^{M}n_j \theta_{ij} + a_{i,I}, \\
&&x_{i,I}^{-}=\alpha_i^* x+\beta_i^* t-\sum_{j=1}^{M}n_j \theta_{ij}^* +a_{i,I}^*,\\
&&\ln \left[\frac{1}{\kappa}\left(\frac{p_{0}+p_{0}^{*}}{p_{1}}\right)\left(\frac{p(\kappa)-p_{0}}{p(\kappa)+p_{0}^{*}}\right)\right]=\sum_{r=1}^{\infty} s_{r} \kappa^{r}, \label{sr}
\end{eqnarray}
where the asterisk `$*$' represents complex conjugation, $p_0=p(0),~p_1=p^{\prime}(0)$, the $a_{i,I}$'s are arbitrary constants, and $\alpha_i,~\beta_i,~\theta_{ij},~j = 1, 2, \dots, M$, are defined by the expansions
$$
p(\kappa)-p_{0}=\sum_{r=1}^{\infty} \alpha_{r} \kappa^{r}, \quad p^{3}(\kappa)-p_{0}^{3}=\sum_{r=1}^{\infty} \beta_{r} \kappa^{r},\quad\ln \left[\frac{p(\kappa)-\mathrm{i} k_{j}}{p_{0}-\mathrm{i} k_{j}}\right]=\sum_{r=1}^{\infty} \theta_{rj} \kappa^{r}.
$$
\end{theorem}

\subsection{Rogue wave solutions associated with special indices}
\label{subsec:index-vectors-special-polynomials}

In suitable large-parameter regimes, the special polynomials governing rogue wave patterns are determined by the row indices in the underlying Schur-polynomial determinant. An index jump of two gives rise to the Yablonskii--Vorob'ev polynomial hierarchy, whereas an index jump of three leads to the Okamoto polynomial hierarchies \cite{yang2021rogue,yang2021universal,yang2023rogue}. More general index choices can produce broader families of generalized Wronskian--Hermite polynomials \cite{zhang2022rogue,wu2026zerosgeneralizedwronskianhermitepolynomials}.

For either determinant solution in the preceding subsections, the
\((I,J)\) block is
\[
\tau_{\mathbf n}^{[I,J]}
=
\left(
m_{(M+1)(i-1)+I,\,(M+1)(j-1)+J}^{(\mathbf n,I,J)}
\right)_{\substack{1\leq i\leq N_I, \, 1\leq j\leq N_J}}.
\]
Therefore, as the local row number \(i\) runs from
\(1\) to \(N_I\), the row indices contributed by the \(I\)-th block are
\[
I,\ M+1+I,\ \ldots,\ (M+1)(N_I-1)+I.
\]
Let \(I_1<I_2<\cdots<I_K\) be the active block indices. Reading the
row blocks from top to bottom, and the rows within each block in their
natural order, gives the index vector
\[
\boldsymbol{\mu}
=
\bigl(
I_1,M+1+I_1,\ldots,(M+1)(N_{I_1}-1)+I_1;
\ldots;
I_K,M+1+I_K,\ldots,(M+1)(N_{I_K}-1)+I_K
\bigr).
\]
The column blocks produce the same ordered indices. The length of \(\boldsymbol{\mu}\) is
\[
N=N_{I_1}+N_{I_2}+\cdots+N_{I_K}.
\]
This is the index vector used in
\(P_{\boldsymbol{\mu}}\) defined by \eqref{eq:Pmu}.

We first consider consecutive index vectors. Fix a positive integer \(M\). For each pair \((d,n)\) satisfying
\[
d\geq1,\qquad n\geq2,\qquad d+n-1\leq M,
\]
we choose
\begin{equation}
\label{N choose}
\mathcal N=(N_1,N_2,\ldots,N_M),\qquad
N_d=N_{d+1}=\cdots=N_{d+n-1}=1,
\qquad
N_i=0\quad\text{for all other }i.
\end{equation}
For this choice,
\[
\boldsymbol{\mu}=(d,d+1,\ldots,d+n-1),
\]
and hence
\[
P_{\boldsymbol{\mu}}(z)=H_{d,n}(z).
\]
Moreover, the \(\tau\)-functions in \eqref{tau-block matrix-theorem-1} and \eqref{tau-block matrix-theorem-2} reduce to
\begin{equation}
\tau_{\mathbf n}
=
\det
\begin{pmatrix}
m_{d,d} & m_{d,d+1} & \cdots & m_{d,d+n-1}\\
m_{d+1,d} & m_{d+1,d+1} & \cdots & m_{d+1,d+n-1}\\
\vdots & \vdots & \ddots & \vdots\\
m_{d+n-1,d} & m_{d+n-1,d+1} & \cdots & m_{d+n-1,d+n-1}
\end{pmatrix}_{n\times n},
\end{equation}
where \(m_{i,j}\) is defined in the same way as \(m_{i,j}^{(\mathbf n,I,J)}\). We call the resulting solutions consecutive-index rogue wave solutions.

For fixed \(M\), each admissible pair \((d,n)\) specifies a consecutive-index solution family associated with \(H_{d,n}(z)\), and
there are \(M(M-1)/2\) such families. Conversely, every generalized Hermite polynomial \(H_{d,n}(z)\) with \(d\geq1\) and \(n\geq2\) can be realized by this construction. Its smallest admissible number of components is \(M_0=d+n-1\). For every \(M\geq M_0\), the choice \eqref{N choose}, with all additional entries set to zero, yields a solution family associated with \(H_{d,n}(z)\). Thus, \(H_{d,n}(z)\) first occurs when \(M=M_0\) and continues to occur for every \(M>M_0\).

The second specialization used below occurs when \(M=2\). For \(\mathcal N=(N_1,N_2)\), the associated index vector is
\[
\boldsymbol{\mu}
=
\bigl(
1,4,\ldots,3N_1-2;
2,5,\ldots,3N_2-1
\bigr),
\]
where the first or second segment is omitted when \(N_1=0\) or \(N_2=0\), respectively. Consequently,
\[
P_{\boldsymbol{\mu}}(z)=Q_{N_1,N_2}(z).
\]
Thus, the generalized Hermite and generalized Okamoto polynomials used later arise from two different choices of the same determinant index vector.

\section{Rogue wave patterns associated with generalized Hermite and generalized Okamoto polynomials}
\label{sec:patterns}

\subsection{Notations}
\label{Notations}
For every active block, take a common second internal parameter and write
\begin{equation}
\label{eq:large-parameter}
a_{2,I}=a_2=B^2,\qquad
a_{r,I}=0\quad(r\ne2),\qquad
\abs{B}\to\infty.
\end{equation}
A branch of \(B=\sqrt{a_2}\) is fixed throughout. For the multi-component NLS equation, put
\begin{equation}
\label{eq:xi}
 \xi_r(x,t)=\alpha_rx+\i \beta_rt,
 \qquad
 \xi_1(x,t)=p_1(x+2\i  p_0t).
\end{equation}
For the multi-component Hirota equation, put
\begin{equation}
\label{eq:eta}
\eta_r(x,t)=\alpha_rx+\gamma_rt,
\qquad
\gamma_r=\beta_r-3\alpha_r\sum_{j=1}^{M}c_j\rho_j^2.
\end{equation}
In particular,
\begin{equation}
\label{eq:eta1}
\eta_1(x,t)=p_1(x+3qt),
\qquad
q=p_0^2-\sum_{j=1}^{M}c_j\rho_j^2.
\end{equation}
The definitions of \(\gamma_r\) and \(\eta_r\) absorb the coordinate
shift \(x\mapsto x-3\sum_{j=1}^{M}c_j\rho_j^2t\) appearing in the rogue wave solution of the multi-component Hirota equation. Thus, throughout the remainder of this
paper, the tau-function is written simply as
\(\tau_{\bm n}(x,t)\).

For both equations, we also set 
\[
h_0=\frac{\abs{p_1}}{p_0+p_0^*}.
\]

The conditions $\Re(p_0)\ne0$ for the multi-component NLS equation and $\Re(p_0)\Im(p_0)\ne0$ for the Hirota equation ensure that, for every $w\in\mathbb C$, the equations $\xi_1(x,t)=w$ and $\eta_1(x,t)=w$, respectively, have unique solutions $(x,t)\in\mathbb R^2$.

Let \(P_{\boldsymbol{\mu}}\) be defined by \eqref{eq:Pmu}. For \(\ell=1,\ldots,N\), let \(P_{\boldsymbol\mu}^{[\ell]}(z)\) be the determinant obtained from \eqref{eq:Pmu} by replacing every \(p_k(z)\) in its \(\ell\)-th column with \(p_{k-2}(z)\), while keeping the same normalization factor \(C_{\boldsymbol{\mu}}\).

\begin{lemma}
\label{lem:Euler-correction}
Let \(\Lambda=\deg P_{\boldsymbol\mu}(z)\). Then
\begin{equation}
\label{eq:Euler-correction}
2\sum_{\ell=1}^N P_{\boldsymbol\mu}^{[\ell]}(z)=\Lambda P_{\boldsymbol\mu}(z)-zP_{\boldsymbol\mu}'(z).
\end{equation}
Consequently, at any simple root \(z_0\) of \(P_{\boldsymbol\mu}\),
\begin{equation}
\label{eq:root-correction-ratio}
\frac{\displaystyle\sum_{\ell=1}^N P_{\boldsymbol\mu}^{[\ell]}(z_0)}{P_{\boldsymbol\mu}'(z_0)}=-\frac{z_0}{2}.
\end{equation}
\end{lemma}

\begin{proof}
Introduce \(p_k(z,y)\) by
\begin{equation}
\sum_{k\geq0}p_k(z,y)\epsilon^k=\exp\left(z\epsilon+y\epsilon^2\right),
\end{equation}
with \(p_k(z,y)\equiv0\) for \(k<0\), and, using \(C_{\boldsymbol\mu}\) given by \eqref{eq:Wronskian-normalization}, define
\begin{equation}
P_{\boldsymbol\mu}(z,y)=C_{\boldsymbol\mu}\det\left(p_{\mu_i-j+1}(z,y)\right)_{i,j=1}^N.
\end{equation}
It follows from this definition and \eqref{eq:Pmu-degree} that
\begin{equation}
P_{\boldsymbol\mu}(\lambda z,\lambda^2y)=\lambda^\Lambda P_{\boldsymbol\mu}(z,y).
\end{equation}
Euler's identity gives
\begin{equation}
z\partial_zP_{\boldsymbol\mu}(z,y)+2y\partial_yP_{\boldsymbol\mu}(z,y)=\Lambda P_{\boldsymbol\mu}(z,y).
\end{equation}
Since \(\partial_y p_k=p_{k-2}\), we have
\begin{equation}
\partial_yP_{\boldsymbol\mu}(z,1)=\sum_{\ell=1}^N P_{\boldsymbol\mu}^{[\ell]}(z).
\end{equation}
Setting \(y=1\) proves \eqref{eq:Euler-correction}, and evaluating it at a simple root proves \eqref{eq:root-correction-ratio}.
\end{proof}

For a root $z_0$ of $P_{\boldsymbol\mu}$, let $(\widetilde x_0,\widetilde t_0) \in \R^2$ be the point satisfying
\begin{equation}
 \xi_1(\widetilde x_0,\widetilde t_0)=Bz_0
\end{equation}
for the multi-component NLS equation, and define
\begin{equation}
\label{eq:Delta}
 \Delta_{P_{\boldsymbol\mu}}(z_0;B)=\frac{\xi_2(\widetilde x_0,\widetilde t_0)}{B}
 \frac{\displaystyle\sum_{\ell=1}^{N}P_{\boldsymbol\mu}^{[\ell]}(z_0)}{P_{\boldsymbol\mu}'(z_0)}=-\frac{z_0\xi_2(\widetilde x_0,\widetilde t_0)}{2B}.
\end{equation}
For each root $z_0$ of $P_{\boldsymbol\mu}$, the predicted center of the NLS rogue wave is determined by
\begin{equation}
\label{eq:NLS-center-implicit}
 \xi_1(\widehat x_0,\widehat t_0)=Bz_0-\Delta_{P_{\boldsymbol\mu}}(z_0;B),
\end{equation}
which yields
\begin{equation}
\label{eq:NLS-center}
\widehat x_0=\frac1{\Re p_0}\Re\left[\frac{p_0^*}{p_1}\bigl(Bz_0-\Delta_{P_{\boldsymbol\mu}}(z_0;B)\bigr)\right], \quad
\widehat t_0=\frac1{2\Re p_0}\Im\left[\frac{Bz_0-\Delta_{P_{\boldsymbol\mu}}(z_0;B)}{p_1}\right].
\end{equation}

For the Hirota equation, define $(\widetilde x_0,\widetilde t_0)$ by
$\eta_1(\widetilde x_0,\widetilde t_0)=Bz_0$ and set
\begin{equation}
\label{eq:Gamma}
 \Gamma_{P_{\boldsymbol\mu}}(z_0;B)=\frac{\eta_2(\widetilde x_0,\widetilde t_0)}{B}
 \frac{\displaystyle\sum_{\ell=1}^{N}P_{\boldsymbol\mu}^{[\ell]}(z_0)}{P_{\boldsymbol\mu}'(z_0)}=-\frac{z_0\eta_2(\widetilde x_0,\widetilde t_0)}{2B}.
\end{equation}
For each root $z_0$ of $P_{\boldsymbol\mu}$, the predicted center of the Hirota rogue wave satisfies
\begin{equation}
\label{eq:Hirota-center-implicit}
 \eta_1(\widehat x_0,\widehat t_0)=Bz_0-\Gamma_{P_{\boldsymbol\mu}}(z_0;B),
\end{equation}
and hence
\begin{equation}
\label{eq:Hirota-center}
\widehat x_0=\frac1{\Im q}\Im\left[q\left(\frac{Bz_0-\Gamma_{P_{\boldsymbol\mu}}(z_0;B)}{p_1}\right)^*\right], \quad
\widehat t_0=\frac1{3\Im q}\Im\left[\frac{Bz_0-\Gamma_{P_{\boldsymbol\mu}}(z_0;B)}{p_1}\right].
\end{equation}
For fixed $z_0$, the defining equations for $(\widetilde x_0,\widetilde t_0)$ imply that $(\widetilde x_0,\widetilde t_0)=O(\abs{B})$. Since $\xi_2$ and $\eta_2$ are linear in $(x,t)$, it follows from \eqref{eq:Delta} and \eqref{eq:Gamma} that $\Delta_{P_{\boldsymbol\mu}}(z_0;B)=O(1)$ and $\Gamma_{P_{\boldsymbol\mu}}(z_0;B)=O(1)$ as $\abs{B}\to\infty$. If $z_0=0$, then $(\widetilde x_0,\widetilde t_0)=(0,0)$ and hence $\Delta_{P_{\boldsymbol\mu}}(0;B)=\Gamma_{P_{\boldsymbol\mu}}(0;B)=0$. Consequently, \eqref{eq:NLS-center-implicit} and \eqref{eq:Hirota-center-implicit} give $(\widehat x_0,\widehat t_0)=(0,0)$.

\subsection{Rogue wave patterns associated with generalized Hermite and generalized Okamoto polynomials}

To describe rogue wave patterns, we first define
\begin{equation}
\label{eq:fundamental-profile}
 \mathcal P_j(\zeta)=
 \frac{(\zeta+\theta_{1j})(\zeta^*-\theta_{1j}^*)+h_0^2}
 {\abs{\zeta}^2+h_0^2}.
\end{equation}
The normalized fundamental profiles are
\begin{equation}
\label{eq:fundamental-profiles}
 \widehat u_j(x,t)=\mathcal P_j(\xi_1(x,t)),
 \qquad
 \widehat v_j(x,t)=\mathcal P_j(\eta_1(x,t)).
\end{equation}
They are the lowest-order members of the determinant families in
\eqref{eq:NLS-solution} and \eqref{eq:Hirota-solution}, respectively.

\begin{theorem}
\label{thm:common-splitting}
Consider either of the following two index choices.

\begin{enumerate}
\item (Generalized Hermite polynomials) Let \(d\geq1\), \(n\geq2\), and \(M\geq d+n-1\), and take 
\[
\mathcal N=(N_1,N_2,\ldots,N_M),\qquad
N_d=N_{d+1}=\cdots=N_{d+n-1}=1,
\qquad
N_i=0\quad\text{for all other }i.
\]
For this choice, 
\[
\boldsymbol{\mu}=(d,d+1,\ldots,d+n-1),
\qquad
P_{\boldsymbol{\mu}}=H_{d,n},
\qquad
\Lambda=dn.
\]

\item (Generalized Okamoto polynomials) Let \(M=2\) , and take
\[
\mathcal N=(N_1,N_2),\qquad N_1+N_2\geq2.
\]
For this choice,
\[
\boldsymbol{\mu}
=
\bigl(1,4,\ldots,3N_1-2;
      2,5,\ldots,3N_2-1\bigr),
\qquad
P_{\boldsymbol{\mu}}=Q_{N_1,N_2},
\qquad
\Lambda=d_{N_1,N_2},
\]
where  the first or second segment is omitted when \(N_1=0\) or \(N_2=0\), respectively; \(d_{N_1,N_2}\) is given by
\eqref{eq:degree-generalized-Okamoto}.
\end{enumerate}

Let \(z_1,\ldots,z_{\Lambda}\) be the roots of \(P_{\boldsymbol\mu}\). For the
multi-component NLS equation, let \(\tau_{\bm n}(x,t)\) denote the
determinant defined in \eqref{tau-block matrix-theorem-1}, set
\(\chi_r=\xi_r\), and define the predicted centers by
\eqref{eq:NLS-center}. For the multi-component Hirota equation, let
\(\tau_{\bm n}(x,t)\) denote the determinant defined in
\eqref{tau-block matrix-theorem-2}, set \(\chi_r=\eta_r\), and define
the predicted centers by \eqref{eq:Hirota-center}. For both equations,
we impose the common large-parameter choice
\eqref{eq:large-parameter}.

For each \(\alpha=1,\ldots,\Lambda\), when $\sqrt{(x-\widehat x_\alpha)^2+(t-\widehat t_\alpha)^2}=O(1)$, we have
\begin{equation}
\label{eq:tau-ratio-local}
\frac{\tau_{\bm e_j}(x,t)}{\tau_{\bm0}(x,t)}
=
\mathcal P_j\!\left(
\chi_1(x-\widehat x_\alpha,t-\widehat t_\alpha)
\right)
+O(\abs{B}^{-1}),
\qquad j=1,\ldots,M,\quad \text{as} \ \abs{B}\to\infty.
\end{equation}
Moreover, when $\sqrt{x^2+t^2}=O(|B|)$ or $\sqrt{x^2+t^2}=O(1)$, if $(x,t)$ lies outside the $O(1)$ neighborhoods of all the predicted centers $(\widehat x_\alpha,\widehat t_\alpha)$, then
\begin{equation}
\label{eq:tau-ratio-background}
\frac{\tau_{\bm e_j}(x,t)}{\tau_{\bm 0}(x,t)}
=
1+o(1),
\qquad j=1,\ldots,M.
\end{equation}
\end{theorem}

\begin{corollary}[NLS rogue wave patterns]
\label{cor:NLS-PIV-patterns}
Under either index choice in Theorem~\ref{thm:common-splitting}, consider rogue wave solutions of the multi-component NLS equation. For each root \(z_\alpha\) of \(P_{\boldsymbol\mu}\), let \((\widehat x_\alpha,\widehat t_\alpha)\) be the predicted center given by \eqref{eq:NLS-center}. Then, as \(\abs{B}\to\infty\), the $\mathcal N$-th order rogue wave solution splits into $\Lambda$ fundamental rogue waves. More precisely, in an \(O(1)\)-neighborhood of each predicted center,
\begin{equation}
\label{eq:NLS-local-asymptotic}
u_{j,\mathcal N}(x,t)
=
\rho_j\widehat u_j(x-\widehat x_\alpha,t-\widehat t_\alpha)
e^{\i(k_jx+w_jt)}
+O(\abs{B}^{-1}),
\qquad j=1,\ldots,M.
\end{equation}
When $\sqrt{x^2+t^2}=O(|B|)$ or $\sqrt{x^2+t^2}=O(1)$, the solution $u_{j,\mathcal N}(x,t)$ asymptotically approaches the background outside the $O(1)$ neighborhoods of the predicted centers.
\end{corollary}

\begin{corollary}[Hirota rogue wave patterns]
\label{cor:Hirota-PIV-patterns}
Under either index choice in Theorem~\ref{thm:common-splitting}, consider rogue wave solutions of the multi-component Hirota equation. For each root \(z_\alpha\) of \(P_{\boldsymbol\mu}\), let \((\widehat x_\alpha,\widehat t_\alpha)\) be the predicted center given by \eqref{eq:Hirota-center}. Then, as \(\abs{B}\to\infty\), the $\mathcal{N}$-th order rogue wave solution splits into $\Lambda$ fundamental rogue waves. More precisely, in an
$O(1)$-neighborhood of each predicted center,
\begin{equation}
\label{eq:Hirota-local-asymptotic}
v_{j,\mathcal N}(x,t)
=
\rho_j\widehat v_j(x-\widehat x_\alpha,t-\widehat t_\alpha)
e^{\i(k_jx-w_jt)}
+O(\abs{B}^{-1}),
\qquad j=1,\ldots,M.
\end{equation}
When $\sqrt{x^2+t^2}=O(|B|)$ or $\sqrt{x^2+t^2}=O(1)$, the solution $v_{j,\mathcal N}(x,t)$ asymptotically approaches the background outside the $O(1)$ neighborhoods of the predicted centers.
\end{corollary}

\begin{remark}
Although we establish the generalized Okamoto rogue wave patterns only
for the two-component NLS and Hirota equations, the same connection
extends to \(M\)-component integrable systems whenever the $\tau$ functions of their
rogue wave solutions can be expressed as determinants whose entries
involve Schur polynomials with an index jump of three. For the
\(M\)-component NLS and Hirota equations, such solutions arise by
choosing \(p_0\) as a double root of \(\mathcal G_M'(p)=0\) and
\(\mathcal F_M'(p)=0\), respectively.

\end{remark}

\begin{remark}
For a fixed \(M\), only finitely many generalized Hermite polynomials
can be associated with rogue wave patterns of the \(M\)-component NLS
and Hirota equations through the rogue wave solutions given in
Theorems~\ref{RW solutions of vector NLS}
and~\ref{thm:rogue wave solutions}, respectively.
Specifically, the consecutive-index choice
\[
\boldsymbol{\mu}=(d,d+1,\ldots,d+n-1),
\qquad d\geq1,\quad n\geq2,\quad d+n-1\leq M,
\]
leads to the generalized Hermite polynomial \(H_{d,n}(z)\). Thus, for
a fixed \(M\), there are \(M(M-1)/2\) admissible ordered pairs \((d,n)\).
Conversely, for any generalized Hermite polynomial \(H_{d,n}(z)\),
its associated rogue wave patterns can be realized in both equations
by choosing any sufficiently large number of components
\(M\geq d+n-1\) and taking the corresponding consecutive indices.
\end{remark}

\subsection{Numerical illustrations}

In this subsection, we present several examples to illustrate
Theorem~\ref{thm:common-splitting}.
  We first consider the patterns associated with generalized
Hermite polynomials in the 9-component NLS equation and then those
associated with generalized Okamoto polynomials in the 2-component
Hirota equation. 

In the first example, we take
\begin{equation}
\label{eq:nls-hermite-numerical-parameters}
M=9,\qquad
k_1=\frac{1}{2},\qquad
\lambda_1=1,\qquad
\sigma_j=1\quad (j=1,\ldots,9).
\end{equation}
The remaining \(k_j\) and \(\rho_j\) are determined by
\eqref{contraints of multiple zeros}. According to \eqref{eq:NLS-multiple-root}, we take
\[
p_0
=
\sin\frac{\pi}{10}
-\mathrm{i}\cos\frac{\pi}{10}
+\frac{\mathrm{i}}{2}.
\]
We consider the four cases
\[
(d,n)\in\{(1,6),(5,2),(4,3),(5,5)\}.
\]
For each case, we set
\[
a_{2,I}=B^2=50, \qquad B=\sqrt{50}>0,
\]
set all other internal parameters to
zero, and display
\(\lvert u_{5,\mathcal N}(x,t)\rvert\).

In Figure~\ref{fig:nls-hermite-numerical-confirmation}, the fifth-component amplitudes
\(\lvert u_{5,\mathcal{N}}(x,t)\rvert\) of the four exact rogue wave
solutions are plotted in the third row. For $$\mathcal{N}=(1,1,1,1,1,1,0,0,0),$$ the exact solution contains $6$ rogue waves arranged approximately along a straight line. For $$\mathcal{N}=(0,0,0,0,1,1,0,0,0)\quad \text{and} \quad \mathcal{N}=(0,0,0,1,1,1,0,0,0), $$ the  exact solutions
contain  $10$ and $12$  rogue waves, which approximately form skewed \(5\times2\) and \(4\times3\) arrays, respectively. For $$\mathcal{N}=(0,0,0,0,1,1,1,1,1),$$
the exact solution contains $25$  rogue waves forming a skewed \(5\times5\) array.

Next, we use~\eqref{eq:NLS-center} to predict the locations of the rogue waves in these exact solutions. The generalized Hermite polynomials \(H_{1,6}(z)\), \(H_{5,2}(z)\), \(H_{4,3}(z)\), and \(H_{5,5}(z)\), whose root configurations are shown in the first row of Figure 3, have 6, 10, 12, and 25 roots, respectively. Then, according to Theorem~\ref{thm:common-splitting}, the four true solutions shown in the third row of Figure~\ref{fig:nls-hermite-numerical-confirmation} are predicted to asymptotically split into \(6\), \(10\), \(12\), and \(25\) fundamental vector rogue waves, respectively. These predicted rogue wave locations are plotted in the second row of Figure~\ref{fig:nls-hermite-numerical-confirmation}. A comparison of the second and third rows shows a close correspondence between the predicted locations and the positions of true rogue waves. 

\begin{figure}[H]

\centering
\begingroup

\setlength{\tabcolsep}{0pt}
\renewcommand{\arraystretch}{1}

\newcommand{\figpanel}[1]{%
  \includegraphics[
    width=\linewidth,
    keepaspectratio
  ]{#1}%
}

\newcommand{\hermiteamppanel}[1]{%
  \makebox[\linewidth][c]{%
    \includegraphics[
      width=0.956\linewidth,
      keepaspectratio
    ]{#1}%
  }%
}

\newcommand{\figrowlabel}[1]{%
  \parbox[b][0.156\textwidth][c]{\linewidth}{%
    \makebox[\linewidth][c]{%
      \rotatebox[origin=c]{90}{\footnotesize #1}%
    }%
  }%
}

\begin{tabular}{
@{}
>{\centering\arraybackslash}m{0.028\textwidth}
@{\hspace{0.2mm}}
*{4}{
>{\centering\arraybackslash}m{0.208\textwidth}
@{\hspace{1.2mm}}
}
>{\centering\arraybackslash}m{0.014\textwidth}
@{}
}

&
{\small $(d,n)=(1,6)$}
&
{\small $(d,n)=(5,2)$}
&
{\small $(d,n)=(4,3)$}
&
{\small $(d,n)=(5,5)$}
&
\\[1mm]

\figrowlabel{Roots}
&
\figpanel{
Figures/Numerical\_confirmation/Vector\_NLS\_H\_1\_6\_A\_50\_Hermite\_roots.png
}
&
\figpanel{
Figures/Numerical\_confirmation/Vector\_NLS\_H\_5\_2\_A\_50\_Hermite\_roots.png
}
&
\figpanel{
Figures/Numerical\_confirmation/Vector\_NLS\_H\_4\_3\_A\_50\_Hermite\_roots.png
}
&
\figpanel{
Figures/Numerical\_confirmation/Vector\_NLS\_H\_5\_5\_A\_50\_Hermite\_roots.png
}
&
{\small $t$}
\\[1.2mm]

\figrowlabel{Prediction}
&
\figpanel{
Figures/Numerical\_confirmation/Vector\_NLS\_H\_1\_6\_A\_50\_RW\_locations.png
}
&
\figpanel{
Figures/Numerical\_confirmation/Vector\_NLS\_H\_5\_2\_A\_50\_RW\_locations.png
}
&
\figpanel{
Figures/Numerical\_confirmation/Vector\_NLS\_H\_4\_3\_A\_50\_RW\_locations.png
}
&
\figpanel{
Figures/Numerical\_confirmation/Vector\_NLS\_H\_5\_5\_A\_50\_RW\_locations.png
}
&
{\small $t$}
\\[1.2mm]

\figrowlabel{$\lvert u_{5,\mathcal N}(x,t)\rvert$}
&
\hermiteamppanel{
Figures/Numerical\_confirmation/NLS\_M9\_H1\_6\_u5\_a2\_50\_amplitude.png
}
&
\hermiteamppanel{
Figures/Numerical\_confirmation/NLS\_M9\_H5\_2\_u5\_a2\_50\_amplitude.png
}
&
\hermiteamppanel{
Figures/Numerical\_confirmation/NLS\_M9\_H4\_3\_u5\_a2\_50\_amplitude.png
}
&
\hermiteamppanel{
Figures/Numerical\_confirmation/NLS\_M9\_H5\_5\_u5\_a2\_50\_amplitude.png
}
&
{\small $t$}
\\[-0.5mm]

&
{\small $x$}
&
{\small $x$}
&
{\small $x$}
&
{\small $x$}
&

\end{tabular}
\endgroup

\caption{
Numerical illustration of rogue wave patterns associated with generalized Hermite polynomials in the 9-component NLS equation. The first row shows the roots of $H_{d,n}(z)$ in the complex $z$-plane, the second row shows the corresponding locations in the $(x,t)$-plane predicted by \eqref{eq:NLS-center}, and the third row shows the fifth-component amplitude $\lvert u_{5,\mathcal N}(x,t)\rvert$ of the true rogue wave solutions. The $(x,t)$-ranges in the first row, where
$x=\operatorname{Re}(z)$ and $t=\operatorname{Im}(z)$, are chosen from
left to right as follows: first column: $-6.35\leq x\leq6.35$,
$-4.76\leq t\leq4.76$; second column: $-6.94\leq x\leq6.94$,
$-5.21\leq t\leq5.21$; third column: $-5.42\leq x\leq5.42$,
$-4.06\leq t\leq4.06$; fourth column: $-6.33\leq x\leq6.33$,
$-4.75\leq t\leq4.75$. The physical $(x,t)$-ranges are chosen identically
in the second and third rows as follows: first column:
$-325\leq x\leq325$, $-250\leq t\leq250$; second column:
$-575\leq x\leq575$, $-425\leq t\leq425$; third column:
$-475\leq x\leq475$, $-350\leq t\leq350$; fourth column:
$-575\leq x\leq575$, $-425\leq t\leq425$.
}
\label{fig:nls-hermite-numerical-confirmation}

\end{figure}

In the second example, we take
\begin{equation}\label{eq:hirota-okamoto-numerical-parameters} M=2,\qquad k_1=1,\qquad \lambda_1=1,\qquad c_j=-1\quad (j=1,2). \end{equation}
The remaining \(k_j\) and \(\rho_j\) are determined by \eqref{eqn:constraints of multiple zeros}, which gives
\[k_2=0,\qquad \rho_1=\rho_2=1.\]
According to \eqref{root of multiplicity M}, we take
\[p_0=\sin\frac{\pi}{3}-\mathrm{i}\cos\frac{\pi}{3}+\mathrm{i}=\frac{\sqrt{3}}{2}+\frac{\mathrm{i}}{2}.\]
We consider the four order-index choices
\[\mathcal N=(N_1,N_2)\in\{(1,2),(2,2),(3,3),(4,4)\}.\]
For each case, we set
\[a_2=B^2=50,\qquad B=\sqrt{50}>0,\]
and all other internal parameters to zero, and display \(\lvert v_{1,\mathcal N}(x,t)\rvert\).

\begin{figure}[H]
\centering
\begingroup

\setlength{\tabcolsep}{0pt}
\renewcommand{\arraystretch}{1}

\newcommand{\figpanel}[1]{%
  \includegraphics[
    width=\linewidth,
    keepaspectratio
  ]{#1}%
}

\newcommand{\okamotoamppanel}[1]{%
  \makebox[\linewidth][c]{%
    \includegraphics[
      width=0.968\linewidth,
      keepaspectratio
    ]{#1}%
  }%
}

\newcommand{\figrowlabel}[1]{%
  \parbox[b][0.156\textwidth][c]{\linewidth}{%
    \makebox[\linewidth][c]{%
      \rotatebox[origin=c]{90}{\footnotesize #1}%
    }%
  }%
}

\begin{tabular}{
@{}
>{\centering\arraybackslash}m{0.028\textwidth}
@{\hspace{0.2mm}}
*{4}{
>{\centering\arraybackslash}m{0.208\textwidth}
@{\hspace{1.2mm}}
}
>{\centering\arraybackslash}m{0.014\textwidth}
@{}
}

&
{\small $(N_1,N_2)=(1,2)$}
&
{\small $(N_1,N_2)=(2,2)$}
&
{\small $(N_1,N_2)=(3,3)$}
&
{\small $(N_1,N_2)=(4,4)$}
&
\\[1mm]

\figrowlabel{Roots}
&
\figpanel{
Figures/Numerical\_confirmation/Hirota\_Q\_1\_2\_a2\_50\_Okamoto\_roots.png
}
&
\figpanel{
Figures/Numerical\_confirmation/Hirota\_Q\_2\_2\_a2\_50\_Okamoto\_roots.png
}
&
\figpanel{
Figures/Numerical\_confirmation/Hirota\_Q\_3\_3\_a2\_50\_Okamoto\_roots.png
}
&
\figpanel{
Figures/Numerical\_confirmation/Hirota\_Q\_4\_4\_a2\_50\_Okamoto\_roots.png
}
&
{\small $t$}
\\[1.2mm]

\figrowlabel{Prediction}
&
\figpanel{
Figures/Numerical\_confirmation/Hirota\_Q\_1\_2\_a2\_50\_Hirota\_corrected\_centers.png
}
&
\figpanel{
Figures/Numerical\_confirmation/Hirota\_Q\_2\_2\_a2\_50\_Hirota\_corrected\_centers.png
}
&
\figpanel{
Figures/Numerical\_confirmation/Hirota\_Q\_3\_3\_a2\_50\_Hirota\_corrected\_centers.png
}
&
\figpanel{
Figures/Numerical\_confirmation/Hirota\_Q\_4\_4\_a2\_50\_Hirota\_corrected\_centers.png
}
&
{\small $t$}
\\[1.2mm]

\figrowlabel{$\lvert v_{1,\mathcal N}(x,t)\rvert$}
&
\okamotoamppanel{
Figures/Numerical\_confirmation/Hirota\_Q1\_2\_a2\_50\_v1\_amplitude.png
}
&
\okamotoamppanel{
Figures/Numerical\_confirmation/Hirota\_Q2\_2\_a2\_50\_v1\_amplitude.png
}
&
\okamotoamppanel{
Figures/Numerical\_confirmation/Hirota\_Q3\_3\_a2\_50\_v1\_amplitude.png
}
&
\okamotoamppanel{
Figures/Numerical\_confirmation/Hirota\_Q4\_4\_a2\_50\_v1\_amplitude.png
}
&
{\small $t$}
\\[-0.5mm]

&
{\small $x$}
&
{\small $x$}
&
{\small $x$}
&
{\small $x$}
&

\end{tabular}
\endgroup

\caption{
Numerical illustration of rogue wave patterns associated with generalized Okamoto polynomials in the 2-component Hirota equation. The first row shows the roots of $Q_{N_1,N_2}(z)$ in the complex $z$-plane, the second row shows the corresponding locations in the $(x,t)$-plane predicted by \eqref{eq:Hirota-center}, and the third row shows the first-component amplitude $\lvert v_{1,\mathcal N}(x,t)\rvert$ of the true rogue wave solutions. The $(x,t)$-ranges in the first row, where $x=\operatorname{Re}(z)$ and $t=\operatorname{Im}(z)$, are chosen from left to right as follows: first column: $-3.05\leq x\leq3.05$, $-2.28\leq t\leq2.28$; second column: $-3.08\leq x\leq3.08$, $-2.31\leq t\leq2.31$; third column: $-4.19\leq x\leq4.19$, $-3.14\leq t\leq3.14$; fourth column: $-5.09\leq x\leq5.09$, $-3.82\leq t\leq3.82$. The physical $(x,t)$-ranges are chosen identically in the second and third rows as follows: first column: $-75\leq x\leq75$, $-11\leq t\leq11$; second column: $-45\leq x\leq45$, $-5\leq t\leq5$; third column: $-60\leq x\leq60$, $-8\leq t\leq8$; fourth column: $-75\leq x\leq75$, $-10\leq t\leq10$.
}
\label{fig:hirota-okamoto-numerical-confirmation}

\end{figure}

In Figure~\ref{fig:hirota-okamoto-numerical-confirmation}, the first-component amplitudes \(\lvert v_{1,\mathcal{N}}(x,t)\rvert\) of the four exact rogue wave solutions are plotted in the third row. For \(\mathcal{N}=(1,2)\), the  exact solution contains \(5\) well-separated rogue waves forming an inclined cross-shaped configuration. For \(\mathcal{N}=(2,2)\), \(\mathcal{N}=(3,3)\), and \(\mathcal{N}=(4,4)\), the exact solutions contain \(6\), \(12\), and \(20\) well-separated rogue waves, respectively. In each of these three cases, the rogue waves form an inclined cone-like configuration in the \((x,t)\)-plane.

For the Hirota equation, we apply~\eqref{eq:Hirota-center} to predict the locations of the rogue waves. The root configurations of the generalized Okamoto polynomials \(Q_{1,2}(z)\), \(Q_{2,2}(z)\), \(Q_{3,3}(z)\), and \(Q_{4,4}(z)\), displayed in the first row of Figure~\ref{fig:hirota-okamoto-numerical-confirmation}, contain \(5\), \(6\), \(12\), and \(20\) roots, respectively. It then follows from Theorem~\ref{thm:common-splitting} that the four true solutions displayed in the third row of Figure~\ref{fig:hirota-okamoto-numerical-confirmation} are expected to decompose asymptotically into \(5\), \(6\), \(12\), and \(20\) fundamental vector rogue waves, respectively. The predicted locations are displayed in the second row of Figure~\ref{fig:hirota-okamoto-numerical-confirmation}. A comparison with the corresponding rogue wave locations in the true solutions reveals excellent agreement. 

\section{Proof of Theorem~\ref{thm:common-splitting}}
\label{sec:proof}

We first recall the following lemma.

\begin{lemma}[{\cite[Lemma~6.3]{wu2026zerosgeneralizedwronskianhermitepolynomials}}]
\label{lem:s1-zero}
Under the assumptions of either Theorem~\ref{RW solutions of vector NLS}
or Theorem~\ref{thm:rogue wave solutions}, the coefficient \(s_1\)
defined in the corresponding theorem satisfies
\[
s_1=0.
\]
\end{lemma}

We now prove Theorem~\ref{thm:common-splitting}. For each \(i\), let \(b_i\) denote the index of the block containing
\(\mu_i\), and set
\[
\mu_{\max}=\max_{1\leq i\leq N}\mu_i.
\]

For either index choice, applying the Cauchy--Binet formula to
\(\tau_{\bm n}\), we obtain
\begin{equation}
\label{eq:Cauchy-Binet}
\tau_{\bm n}
=
\sum_{0\leq\nu_1<\cdots<\nu_N\leq\mu_{\max}}
D_{\boldsymbol\nu}^+(\bm n)
D_{\boldsymbol\nu}^-(\bm n),
\end{equation}
where 
\begin{align}
D_{\boldsymbol\nu}^+(\bm n)
&=\det_{1\leq i,j\leq N}\left[
(h_0)^{\nu_j}S_{\mu_i-\nu_j}
\bigl(\bm x_{b_i}^+(\bm n)+\nu_j\bm s\bigr)\right],
\label{eq:Dplus}\\
D_{\boldsymbol\nu}^-(\bm n)
&=\det_{1\leq i,j\leq N}\left[
(h_0^*)^{\nu_j}S_{\mu_i-\nu_j}
\bigl(\bm x_{b_i}^-(\bm n)+\nu_j\bm s^*\bigr)\right].
\label{eq:Dminus}
\end{align}

For \((x,t)\) in the region
\(\sqrt{x^2+t^2}=O(\abs B)\), we introduce the scaled variable $z=B^{-1}\chi_1(x,t)$. Using the generating function of the Schur polynomials together with
\eqref{eq:large-parameter}, we obtain
\begin{equation}
\label{eq:outer-Schur}
 S_k\left(\bm x_I^+(\bm n)+\nu\bm s\right)
 =B^k\left[p_k(z)+O(\abs B^{-1})\right].
\end{equation}
Consequently,
\begin{equation}
\label{eq:Dnu-outer}
 D_{\boldsymbol\nu}^+(\bm n)
 =h_0^{\sum_j\nu_j}B^{\sum_i\mu_i-\sum_j\nu_j}
 \left[\det_{1\leq i,j\leq N}[p_{\mu_i-\nu_j}(z)]+O(\abs B^{-1})\right].
\end{equation}
The leading-order contribution in \eqref{eq:Cauchy-Binet} comes from the
unique index choice
\begin{equation}
\boldsymbol\nu^{(0)}=(0,1,\ldots,N-1).
\end{equation}
For this index choice, using \eqref{eq:Pmu} and \eqref{eq:Pmu-degree}, \eqref{eq:Dnu-outer} becomes
\begin{equation}
\label{eq:D0-outer}
D_{\boldsymbol\nu^{(0)}}^+(\bm n)
=
C_{\boldsymbol\mu}^{-1}h_0^{N(N-1)/2}B^\Lambda
\bigl[P_{\boldsymbol\mu}(z)+O(\abs B^{-1})\bigr].
\end{equation}
Similarly, 
\begin{equation}
\label{eq:D0minus-outer}
D_{\boldsymbol\nu^{(0)}}^-(\bm n)
=
C_{\boldsymbol\mu}^{-1}h_0^{N(N-1)/2}(B^*)^\Lambda
\bigl[P_{\boldsymbol\mu}(z^*)+O(\abs B^{-1})\bigr].
\end{equation}
Thus, substituting \eqref{eq:D0-outer} and
\eqref{eq:D0minus-outer} into \eqref{eq:Cauchy-Binet}, we obtain
\begin{equation}
\label{eq:tau-away}
\tau_{\bm n}
=
\abs{C_{\boldsymbol\mu}}^{-2}
h_0^{N(N-1)}\abs B^{2\Lambda}\abs{P_{\boldsymbol\mu}(z)}^2
\bigl[1+O(\abs B^{-1})\bigr],
\end{equation}
provided that \(z\) stays away from the roots of \(P_{\boldsymbol\mu}\). Since this
leading-order asymptotics is independent of \(\bm n\), it follows that
\[
\frac{\tau_{\bm e_j}}{\tau_{\bm0}}
=
1+O(\abs B^{-1}),
\qquad j=1,\ldots,M.
\]
Thus, the rogue wave solution approaches the
background, except when \(z\) is near a root \(z_0\) of \(P_{\boldsymbol\mu}\), where
the leading term in \eqref{eq:tau-away} vanishes.

Next, we consider \((x,t)\) in an \(O(1)\) neighborhood of a point
\((\widetilde x_0,\widetilde t_0)\) satisfying
\[
\chi_1(\widetilde x_0,\widetilde t_0)=Bz_0,
\]
where \(z_0\) is a root of \(P_{\boldsymbol\mu}\). For \((x,t)\) in this neighborhood, define \(L\) by
\begin{equation}
L
=
\chi_1(x-\widetilde x_0,t-\widetilde t_0)
+\sum_{j=1}^{M}n_j\theta_{1j},
\end{equation}
and set
\[
\bm v=(Bz_0+L,B^2,0,0,\ldots).
\]
Then, for \((x,t)\) in this neighborhood, since \(s_1=0\) and
\(\chi_2(\widetilde x_0,\widetilde t_0)=O(\abs B)\), we have
\begin{align}
S_k\bigl(\bm x_I^+(\bm n)+\nu\bm s\bigr)
={}&
S_k(\bm v)
+\chi_2(\widetilde x_0,\widetilde t_0)S_{k-2}(\bm v)
+O(\abs B^{k-2}) \notag\\
={}&
B^kp_k\left(z_0+\frac{L}{B}\right)
+\chi_2(\widetilde x_0,\widetilde t_0)
 B^{k-2}p_{k-2}(z_0)
+O(\abs B^{k-2}),
\label{eq:local-Schur}
\end{align}
where
\[
S_j(\bm v)
=
B^jp_j\left(z_0+\frac{L}{B}\right).
\]

We now derive the leading-order terms in the Cauchy--Binet expansion
\eqref{eq:Cauchy-Binet} in this neighborhood. These terms come from two
index choices, the first being
\[
\boldsymbol\nu^{(0)}=(0,1,\ldots,N-1),
\]
and the second being
\[
\boldsymbol\nu^{(1)}=(0,1,\ldots,N-2,N).
\]

With the first index choice $\boldsymbol\nu^{(0)}$, applying \eqref{eq:local-Schur} to
\eqref{eq:Dplus} gives
\begin{align}
D_{\boldsymbol\nu^{(0)}}^+(\bm n)
={}&
C_{\boldsymbol\mu}^{-1}h_0^{N(N-1)/2}
B^{\Lambda-1}P_{\boldsymbol\mu}'(z_0)
\left[
L+
\frac{\chi_2(\widetilde x_0,\widetilde t_0)}{B}
\frac{\displaystyle\sum_{\ell=1}^{N}P_{\boldsymbol\mu}^{[\ell]}(z_0)}
     {P_{\boldsymbol\mu}'(z_0)}
+O(\abs B^{-1})
\right],
\end{align}
which can be rewritten as
\begin{align}
\label{eq:D0-local}
D_{\boldsymbol\nu^{(0)}}^+(\bm n)
={}&
C_{\boldsymbol\mu}^{-1}h_0^{N(N-1)/2}
B^{\Lambda-1}P_{\boldsymbol\mu}'(z_0)
\left[
\chi_1(x-\widehat x_0,t-\widehat t_0)
+\sum_{j=1}^{M}n_j\theta_{1j}
+O(\abs B^{-1})
\right],
\end{align}
where $(\widehat x_0,\widehat t_0)$ is given in \eqref{eq:NLS-center} for the multi-component NLS equation and \eqref{eq:Hirota-center} for the multi-component Hirota equation.
Similarly,
\begin{align}
\label{eq:D0minus-local}
D_{\boldsymbol\nu^{(0)}}^-(\bm n)
={}&
C_{\boldsymbol\mu}^{-1}h_0^{N(N-1)/2}
(B^*)^{\Lambda-1}\bigl[P_{\boldsymbol\mu}'(z_0)\bigr]^*
\left[
\chi_1(x-\widehat x_0,t-\widehat t_0)^*
-\sum_{j=1}^{M}n_j\theta_{1j}^*
+O(\abs B^{-1})
\right].
\end{align}

Under the second index choice $\boldsymbol\nu^{(1)}$, the leading-order term is
\begin{equation}
\label{eq:D1-local}
D_{\boldsymbol\nu^{(1)}}^+(\bm n)
=
C_{\boldsymbol\mu}^{-1}h_0^{N(N-1)/2+1}
B^{\Lambda-1}P_{\boldsymbol\mu}'(z_0)
\bigl[1+O(\abs B^{-1})\bigr].
\end{equation}
Similarly,
\begin{equation}
\label{eq:D1minus-local}
D_{\boldsymbol\nu^{(1)}}^-(\bm n)
=
C_{\boldsymbol\mu}^{-1}h_0^{N(N-1)/2+1}
(B^*)^{\Lambda-1}\bigl[P_{\boldsymbol\mu}'(z_0)\bigr]^*
\bigl[1+O(\abs B^{-1})\bigr].
\end{equation}

Substituting these leading-order contributions into
\eqref{eq:Cauchy-Binet}, we obtain
\begin{equation}
\label{eq:tau-local}
\begin{aligned}
\tau_{\bm n}
={}&
\abs{C_{\boldsymbol\mu}}^{-2}
h_0^{N(N-1)}
\abs B^{2\Lambda-2}
\abs{P_{\boldsymbol\mu}'(z_0)}^2\\
&\times
\left\{
\left[
\chi_1(x-\widehat x_0,t-\widehat t_0)
+\sum_{j=1}^{M}n_j\theta_{1j}
\right]
\left[
\chi_1(x-\widehat x_0,t-\widehat t_0)^*
-\sum_{j=1}^{M}n_j\theta_{1j}^*
\right]
+h_0^2+O(\abs B^{-1})
\right\}.
\end{aligned}
\end{equation}
Since \(z_0\) is a simple root of \(P_{\boldsymbol\mu}\), \(P_{\boldsymbol\mu}'(z_0)\neq0\). It is then easy to see that \eqref{eq:tau-local} yields
\eqref{eq:NLS-local-asymptotic} and
\eqref{eq:Hirota-local-asymptotic}. 

Finally, we consider the solution in the region
\[
x^2+t^2=O(1).
\]
In this region,
\[
z=B^{-1}\chi_1(x,t)=O(\abs B^{-1})\to0.
\]
By the generating function of the Schur polynomials and
\eqref{eq:large-parameter}, the asymptotic expansions
\eqref{eq:outer-Schur} and \eqref{eq:D0-outer}
also hold in this region.

If zero is not a root of \(P_{\boldsymbol\mu}\), then
\[
P_{\boldsymbol\mu}(z)=P_{\boldsymbol\mu}(0)+O(\abs B^{-1}),
\qquad P_{\boldsymbol\mu}(0)\neq0.
\]
Therefore, \(z\) remains separated from the roots of \(P_{\boldsymbol\mu}\), and
\eqref{eq:tau-away} reduces to
\begin{equation}
\tau_{\bm n}
=
\abs{C_{\boldsymbol\mu}}^{-2}
h_0^{N(N-1)}
\abs B^{2\Lambda}
\abs{P_{\boldsymbol\mu}(0)}^2
\bigl[1+O(\abs B^{-1})\bigr].
\end{equation}
The leading-order term is independent of \(\bm n\), and hence
\begin{equation}
    \frac{\tau_{\bm e_j}}{\tau_{\bm0}}
=
1+O(\abs B^{-1}),
\qquad j=1,\ldots,M.
\end{equation}
Thus, the rogue wave solution approaches the background.

If zero is a root of \(P_{\boldsymbol\mu}\), we have
\[
(\widetilde x_0,\widetilde t_0)
=
(\widehat x_0,\widehat t_0)
=
(0,0).
\]
Since \(\chi_2(0,0)=0\), equation \eqref{eq:local-Schur} reduces to
\[
S_k\bigl(\bm x_I^+(\bm n)+\nu\bm s\bigr)
=
B^kp_k\left(\frac{L}{B}\right)
+O(\abs B^{k-2}).
\]
Evaluating \eqref{eq:D0-local}--\eqref{eq:tau-local} at \(z_0=0\),
we find that the solution approaches a fundamental vector rogue wave
centered at the origin, with the asymptotics given by
\eqref{eq:NLS-local-asymptotic} and
\eqref{eq:Hirota-local-asymptotic}. This completes the proof.

\section{Conclusion}
\label{sec:conclusion}

In this paper, we have established the connection between large parameter asymptotics of higher-order rogue wave solutions of the multi-component NLS and Hirota equations and two families of special polynomials associated with the fourth Painlev\'e equation. Using a common Schur-polynomial determinant framework, we have shown that consecutive determinant indices lead to the generalized Hermite polynomials, whereas the index with jumps of three leads to the generalized Okamoto polynomials. When one of the internal parameters is large, each root of the governing polynomial corresponds to one fundamental vector rogue wave. We derived explicit formulas for its space--time center. It is shown that, near each predicted center, a fundamental rogue wave can be found and, away from these centers, the solution approaches the background. The numerical examples for the multi-component NLS and Hirota equations agree with these predictions.

Although rogue wave patterns have been studied in many integrable
systems \cite{yang2021rogue,yang2021universal,yang2023rogue,
zhang2022rogue,lin2024rogue,yang2024rogue,zhang2025multi,yang2026roguewavelumppatternsassociated,ling2024fokaslenells}, a rigorous large-parameter analysis of such patterns for the
Sasa--Satsuma and coupled Sasa--Satsuma equations is still lacking.
Both equations can be obtained from the multi-component Hirota equation
through suitable reductions, and their soliton and
rogue wave solutions have been constructed by the KP reduction method
\cite{ohta2010dark,feng2022higher,wu2022general,zhang2025dark,zhang2025rogue}.
Establishing their rogue wave patterns will be an interesting direction
for future research.

Rogue wave patterns associated with \(\mathrm{P}_{\mathrm{II}}\) and \(\mathrm{P}_{\mathrm{III}}\) have already been reported \cite{yang2021universal,yang2026roguewavelumppatternsassociated}. The results of the present paper, together with those of \cite{yang2023rogue}, establish corresponding rogue wave patterns for all special polynomials associated with rational solutions of \(\mathrm{P}_{\mathrm{IV}}\), thereby completing the \(\mathrm{P}_{\mathrm{IV}}\) part of the correspondence between special polynomials arising from rational solutions of Painlev\'{e} equations and rogue wave patterns. Since \(\mathrm{P}_{\mathrm{V}}\) and \(\mathrm{P}_{\mathrm{VI}}\) admit rational solutions as well \cite{Clarkson2023classical,Clarkson2010Painleve}, a natural question is whether analogous patterns associated with these two equations can also be found (partial results associated with \(\mathrm{P}_{\mathrm{V}}\) were obtained in \cite{wu2026zerosgeneralizedwronskianhermitepolynomials}). Addressing this question would further broaden our understanding of the role of rational solutions of Painlev\'e equations in nonlinear wave phenomena.

\section*{Conflict of interests}
The authors have no conflicts to disclose.

\section*{Acknowledgment}
C.~F.~Wu was supported by the National Natural Science Foundation of China (Grant No. 12471077).

\bibliography{references}
\bibliographystyle{elsarticle-num}

\end{document}